\pdfoutput=1
\documentclass[12pt]{article}

\usepackage{geometry}
\usepackage[utf8x]{inputenc}
\usepackage[sectionbib]{natbib}
\usepackage{array,epsfig,fancyhdr,rotating}
\usepackage[]{hyperref}

\usepackage{amsmath}
\usepackage{amssymb}
\usepackage{amsfonts}
\usepackage{multirow,subcaption,color,float,colortbl}
\usepackage{amsthm}
\usepackage{colortbl,tabularx,makecell}
\usepackage[dvipsnames]{xcolor}
\usepackage[hang,flushmargin]{footmisc}
\usepackage{bbm} 
\usepackage{graphicx,adjustbox}
\usepackage{epstopdf}
\usepackage[normalem]{ulem}
\usepackage{enumitem}  

\newcommand{\tcr}[1]{#1}

\newtheorem{theorem}{Theorem}
\newtheorem{lemma}{Lemma}

\theoremstyle{definition}
\newtheorem{definition}{Definition}

\newtheorem{remark}{Remark}
\def\bm#1{\mbox{\boldmath $#1$}}

\def\E{\mathbb E}
\def\AMSE{{\rm AMSE}}

\def\Cov{{\rm Cov}}

\begin{document}
	
	
	\renewcommand{\baselinestretch}{1.5}
	
	
	\markboth{\hfill{\footnotesize\rm JING ZHOU AND HUI ZOU} \hfill}
	{\hfill {\footnotesize\rm TESTING HETEROSCEDASTICITY} \hfill}

	$\ $\par
	
 \thispagestyle{plain}
 
	\fontsize{12}{14pt plus.8pt minus .6pt}\selectfont \vspace{0.8pc}
	\centerline{\large\bf HIGH-DIMENSIONAL NEWEY-POWELL TEST }
	\centerline{\large\bf VIA APPROXIMATE MESSAGE PASSING}
	\vspace{.4cm} 
	\centerline{Jing Zhou \footnote[1]{Department of Mathematics, University of Manchester. Email: jing.zhou@manchester.ac.uk. The manuscript was written during a research visit to the University of Minnesota (Twin Cities). The first revision was completed at the University of East Anglia, and the second at the University of Manchester. } and Hui Zou\footnote[7]{School of Statistics, University of Minnesota (Twin Cities). Email: zouxx019@umn.edu
}} 
	\vspace{.4cm} 
	\centerline{\it University of Manchester and University of Minnesota Twin Cities}
 
	\vspace{.55cm} \fontsize{12}{11.5pt plus.8pt minus.6pt}\selectfont
	

	\begin{quotation}
     \noindent {\it Abstract:}
\tcr{We propose a high-dimensional extension of the heteroscedasticity test proposed in \citet{newey1987asymmetric}. Our test is based on expectile regression in the proportional asymptotic regime where $n/p \to \delta \in (0,1]$. The asymptotic analysis of the test statistic uses the Approximate Message Passing (AMP) algorithm, from which we obtain the limiting distribution of the test and establish its asymptotic power. The numerical performance of the test is validated through an extensive simulation study. As real-data applications, we present the analysis based on ``international economic growth" data \citep{belloni2011inference}, which is found to be homoscedastic, and ``supermarket" data \citep{lan2016testing}, which is found to be heteroscedastic.}
  \vspace{9pt}
		\par
\end{quotation}\par
	
	\def\thefigure{\arabic{figure}}
	\def\thetable{\arabic{table}}
	
	\renewcommand{\theequation}{\thesection.\arabic{equation}}

	\fontsize{12}{14pt plus.8pt minus .6pt}\selectfont










	\section{Introduction}
\tcr{In many modern economic applications, datasets often contain a large number of covariates $p$ relative to the sample size $n$. When performing high-dimensional regression on such data, parameter estimation typically assumes homoscedastic errors for simplicity. However, heteroscedasticity is common in practice, and ignoring it can lead to inefficient estimates and invalid inference. Despite its practical importance, testing for heteroscedasticity remains both underexplored and challenging in high-dimensional settings.
\\
Classical heteroscedasticity tests, such as those proposed in \citet{breusch1979simple, white1980heteroskedasticity,koenker1982robust}, were developed under the assumption of fixed $p$, and are no longer valid when $p > n$, which is the primary focus of this paper. To address this gap, we propose a high-dimensional heteroscedasticity test in the proportional regime where $n/p \to \delta \in (0,1]$. Our approach extends the classical test of \citet{newey1987asymmetric}, which considers the following heteroscedastic linear model
\begin{equation}\label{eq: linear model}
    Y = X\beta_0 + \sigma(X \gamma_0)\varepsilon,
\end{equation}
where $X = (X_1, \ldots, X_p) \in \mathbb{R}^p$, $\beta_0, \gamma_0 \in \mathbb{R}^p$ are unknown parameter vectors, and $Y, \varepsilon \in \mathbb{R}$, with $\varepsilon$ independent of $X$. The function $\sigma(\cdot)$ captures heteroscedasticity, allowing the error variance to vary with the linear combination $X\gamma_0$.
\\
To detect nonconstant variance, we use the expectile as a distributional statistic, evaluated at various levels $\tau \in (0,1)$, defined as
\begin{equation}\label{eq:u_tau definition}
    u_\tau(Y) - \mathbb{E}(Y) = \frac{2\tau - 1}{1 - \tau} \int_{[u_\tau(Y), \infty)} (y - u_\tau(Y)) \, dF_Y(y).
\end{equation}
\\
In model \eqref{eq: linear model}, the $\tau$th expectile of $Y$ given $X$ is
$
u_\tau(Y \mid X) = X\beta_0 + u_\tau(\sigma(X\gamma_0)\varepsilon).
$
A key insight is that under heteroscedasticity, the population parameter $\beta_0$ is affected by an additional contribution from $\sigma(X\gamma_0)$, scaled by the expectile of the error term $\varepsilon$. In contrast, under homoscedasticity, $\beta_0$ remains invariant across different values of $\tau$. This observation becomes even clearer when a local linear approximation of $\sigma(\cdot)$ in \eqref{eq: linear approximation of linear model} is considered, and it motivates a testing framework based on distributional statistics \citep{newey1987asymmetric, koenker1982robust}, where the hypothesis concerns the influence of $\gamma_0$. }

 Compared to previous work  \citep{newey1987asymmetric}, this paper addresses two additional challenges commonly encountered in contemporary datasets: (i) the number of parameters $p$ exceeds the sample size $n$ and diverges at a linear rate: $n/p \rightarrow \delta \in (0,1]$; and (ii) the $p$-vectors $\beta_0$ and $\gamma_0$ are sparse. This linear rate regime was first discussed in the early AMP literature \citet{donoho2009message, donoho2010message} and then formally stated in \citet[Theorem~1]{bayati2011dynamics}. The AMP algorithm has attracted significant attention for the theoretical analysis of regularized M-estimators and low-rank matrix estimation. Subsequent work often separately discusses $\delta \in (0,1]$ for sparse regression with regularization \citep[for example,][]{su2017false,weinstein2017power,weinstein2020power,zhou2020detangling,bradic2016robustness} and $\delta \in (1, \infty)$ for unstructured regression coefficient vectors \citep{donoho2016high,bradic2016robustness,sur2019likelihood,sur2019modern,candes2020phase}. We refer to \citet{feng2022unifying} for a thorough literature review. In addition to the AMP framework, \citet{el2013robust,el2018impact} introduced an alternative way of analyzing the estimators in this $n, p$ proportional regime, with conclusions agreeing with the AMP analysis. Specifically, the analysis addresses unregularized and ridge regression, where a sparse structure is not assumed for the parameter vector.
  In this work, we focus on sparse regression and address the $n/p \to \delta \in (0,1]$ setting.
    
    In the past decade, fundamental research for high-dimensional linear regression has focused on sparse regularization \citep{tibshirani1996regression, candes2007dantzig, zou2006adaptive, fan2008sure, zhang2010nearly}. We refer to \citet{buhlmann2011statistics} and \citet{fan2020statistical} for a systematic review of this topic. Most existing literature in this field assumes homoscedastic regression errors, i.e., $\sigma(\cdot)$ is a constant, as opposed to the fact that heteroscedasticity is commonly observed in real data, see, for example, \citet{wang2012quantile} and \citet{daye2012high} on genomic data and \citet{belloni2014pivotal} on economic data. 

    There are natural solutions to address heteroscedasticity in model estimation. 
    \citet{cattaneo2018inference,jochmans2022heteroscedasticity} addressed this by splitting the $p$ covariate vector $X$ into a fixed low-dimensional subvector $X_s$ of main interest with unknown parameter $\beta_s$ and a nuisance subvector $X_{p-s}$ with nonnegligible size $q =p-s$. Their high dimensionality refers to a setting where $n > p$ but $p$ is considerably larger than in the classic setting, where $p/n \to 0$.
    The asymptotic normality is further derived for $\hat\beta_s$, and the main focus is on heteroscedastic-robust asymptotic covariance estimators. Addressed by \citet[Section~3.2, 3.3]{jochmans2022heteroscedasticity}, the consistency of the heteroscedasticity-robust covariance estimator requires the length of the nuisance parameter to satisfy $\lim_n \sup q/n < 1$, and \tcr{handling} the case of $1/2 <\lim_n \sup q/n < 1$ is already challenging.
    Another approach considers replacing the least-squares loss function with asymmetric loss functions, such as the quantile loss functions \citep{koenker1978regression, koenker_2005} 
    and the expectile loss functions 
    \citep{newey1987asymmetric}. 
    A thorough review of quantile regression can be found in  \citet{koenker2017quantile}. High-dimensional quantile regression, addressing $p > n$ setting, has been studied by several authors \citep{wu2009variable,belloni2011l1,wang2012quantile,zheng2013adaptive,fhdqr,smoothqr}. \citet{gu2016high} studied high-dimensional expectile regression and proposed COSALES, which couples two expectile loss functions at expectile level $0.5$ and another level of interest. By this construction, COSALES can simultaneously estimate parameter vectors $\beta_0$ and $\gamma_0$ in \eqref{eq: linear approximation of linear model}. 
    Additionally, \citet{zhou2021cross} proposed a cross-fitted residual regression for estimating heteroscedasticity in \eqref{eq: linear approximation of linear model}. 
    
  Despite recent efforts in estimating heteroscedasticity through methods such as quantile regression, expectile regression, or cross-refitted regression, testing for the presence of heteroscedasticity remains largely unexplored in the high-dimensional literature. Since quantile and expectile loss functions are robust to heteroscedastic regression errors, using these two asymmetric loss functions to construct heteroscedastic tests is natural and thus taken in this paper. Moreover, as previously mentioned, the Newey-Powell test is shown to be superior to the quantile-based test in low-dimensional settings \citep{newey1987asymmetric}. Therefore, we focus on the high-dimensional Newey-Powell test in this paper. In fact, the contrast between expectile regression and quantile regression is even more significant in high dimensions, owing to the smoother loss function in expectile regression. \citet{gu2016high} shows that fitting a high-dimensional expectile regression model is as simple as fitting Lasso regression. In contrast, serious efforts have been dedicated to handling the nonsmooth loss function in high-dimensional quantile regression, including ADMM \citep{fhdqr} and smoothing approximation \citep{smoothqr}.

    \tcr{
   Given the above considerations, we extend the Newey-Powell test \citep{newey1987asymmetric} to sparse high-dimensional settings where $n/ p \to \delta \in (0, 1]$, $n, p\to \infty$, and $\ell_1$-regularization or other sparse regularization is necessary to obtain good estimators. 
    The heteroscedasticity of regression error is tested via $H_{0,j}: \gamma_{0, j} = 0$ versus $H_{1, j}: \gamma_{0, j} \neq 0$, $j = 1, \ldots, p$. 
    By using the AMP theory, we establish the joint asymptotic normality of multiple expectile estimators, which we then use to construct the test statistic. We analytically examine the size and power of the proposed test.
}
     
    The rest of the paper is organized as follows. In Section~\ref{sec: sec 2}, we first introduce the notation and assumptions for the high-dimensional heteroscedastic linear model in \eqref{eq: linear approximation of linear model}, following the expectile regression in Section~\ref{ssec: expectile}. The AMP algorithm is introduced in Section~\ref{sec: AMP algorithm}, and the main technical results are included in Section~\ref{ssec:state evolution for multiple estimator}. Section~\ref{sec: test} presents the proposed heteroscedasticity test. The analytical expression of the power function is presented in Theorem~\ref{thm:power function of test}. We discuss a practically important decorrelation step for using AMP in applications in Section~\ref{sec:decorrelation}.  Section~\ref{sec:numerical} contains simulation studies and real data examples from \citet{belloni2011inference} and \citet{lan2016testing}. \tcr{Section~\ref{sec: remarks} concludes the paper. The assumptions, proofs of lemmas and theorems, and auxiliary lemmas and definitions are provided in Appendix~\ref{appendix: assumptions}, Appendix~\ref{appendix: proofs}, and Appendix~\ref{appendix: auxiliary}, respectively.}
    
\section{Heteroscedastic linear regression}\label{sec: sec 2}
    We restate the well-studied form of the heteroscedastic model from \citet{koenker1982robust} and \citet{newey1987asymmetric}
     \begin{equation}\label{eq: linear approximation of linear model}
        Y = X\beta_0 + \big(1 + X\gamma_0\big)  \varepsilon,
    \end{equation}
     where $\beta_0, \gamma_0 \in \mathbbm R^p$ are unknown parameter vectors. In addition, we have $n$ independent and identically distributed (i.i.d.) copies of the pair $(Y, X)$ denoted by $(Y_i, X_i), i= 1, \ldots, n$. Further, we denote $n$-vectors $\bm Y = (Y_1, \ldots, Y_n)^\top$, $\bm \varepsilon = (\varepsilon_{1}, \ldots, \varepsilon_{n})^\top$; the design matrix $\bm X = (X_1, \ldots\, X_n)^\top$. The components of the vector $\bm\varepsilon$ satisfy Assumption~\ref{cond:error}.

    \begin{enumerate}[label=(A1)]
      \item  \label{cond:error} 
       $\varepsilon_1, \ldots, \varepsilon_n$ are i.i.d. random variables with mean zero, a finite $(2\kappa - 2)$th moment for $\kappa\ge 2$.
     \end{enumerate}
    The constant $\kappa$ depends on the loss function used to estimate the unknown parameters $\beta_0,\gamma_0$; in our case, the expectile loss function in \eqref{eq: expectile loss}. This statement will become clear in Section~\ref{ssec:state evo} when we discuss the AMP in detail.
    Further, following \citet{bayati2011dynamics}, the empirical measure of a $p$-vector is defined as putting $1/p$ weight on each component of the vector. Then, the unknown parameter vector $\gamma_0$ complies with Assumption~\ref{cond:gamma}.
     \begin{enumerate}[label=(A2)]
     \item \label{cond:gamma} 
     The empirical measure of the components of the $p$-vector $\gamma_0$, when $p\to \infty$, converges weakly to the probability measure of a random variable $\Gamma_0$ with bounded $(2\kappa-2)$th moment for $\kappa \ge 2$. 
    \end{enumerate}
     This assumption precedes a standard Assumption~\ref{cond:beta} on the parameter vector $\beta_0$ in the Approximate Message Passing (AMP) literature. 
     In practice, the proposed test is valid for $\beta$ with finite $p$ components; see Section~\ref{sec: test}, and the hypothesis in \eqref{eq:hypothesis} tests the components $\gamma_j$'s for finite $p$. We also demonstrate applying the proposed test to two real data examples with finite $p$ in Section~\ref{ssec: real data}. 
     \begin{enumerate}[label=(A3)]
     \item \label{cond:beta}
     The empirical measure of the components of the $p$-vector $\beta_0$, when $p\to \infty$, converges weakly to the probability measure of a random variable $B_0$ with bounded $(2\kappa-2)$th moment for $\kappa\ge 2$. 
    \end{enumerate}
     Notice $X\gamma_0 = \sqrt{n} X  \frac{\gamma_0}{\sqrt{n}},$ where
     the scaled design $\sqrt{n} X $ and parameter vector $\frac{\gamma_0}{\sqrt{n}}$ agree with the assumptions
     in \citet{koenker1982robust}.
      Further, the sequence $\gamma_n = \frac{\gamma_0}{\sqrt{n}}$ restricts the calculation to local alternatives and guarantees well approximation of sequences of multiplicative heteroscedastic models. 
    
    Following \citet[Theorem~1(iii)]{newey1987asymmetric}, the $\tau$th expectile of $Y$ given $X$ can be expressed as  
   \begin{equation}\label{eq: conditional expectile of Y given X}
      u_\tau(Y \mid X) = X\beta_0 + (1 + X\gamma_0)u_\tau(\varepsilon) = X \beta_\tau + u_\tau(\varepsilon),
   \end{equation}  
    where $\beta_\tau = \beta_0 + \gamma_0 u_\tau(\varepsilon)$ is the population parameter vector corresponding to the $\tau$th expectile of $Y$, approximated by $X\beta_\tau$. Here, $u_\tau(\varepsilon)$ is the $\tau$th expectile of the error term, as defined in \eqref{eq:u_tau definition}, and will be denoted simply as $u_\tau$ henceforth. The parameter $\beta_\tau$ is fixed and depends on the scalar $u_\tau$. In this spirit, we rewrite the model as follows
    \begin{equation}\label{eq: hetero linear model}
    Y = X \beta + \varepsilon,
    \end{equation}
    where $\beta = \beta_0 + \varepsilon \gamma_0$ and the error $\varepsilon$ satisfies Assumption~\ref{cond:error}. In addition, we assume the ratio $n / p \to (0, 1]$ when $n, p \to \infty$.
    Let $s = \|\beta\|_0$, where $\|\cdot\|_0$ denotes the $\ell_0$-norm of a vector. Obviously, $s \le  \|\beta_0\|_0+\|\gamma_0\|_0$. The way of parametrization in \eqref{eq: hetero linear model} has two reasons: (1) we want to express the model to fit it in the AMP framework for asymptotic analysis, where an independent error term stands alone; (2) the error term $\varepsilon$ determines the location change of $\beta$ from $\beta_0$. 
    Identifying $\gamma_0$ relies on making use of different spots of the error distribution and requires a pair of regularized estimators. This is technically challenging due not only to the non-trivial asymptotic behavior of regularized estimators when $p>n$, but also to the correlation structure among pairs of estimators whose population parameters do not have the same center. 
    

\subsection{Expectile regression under heteroscedasticity}\label{ssec: expectile}
    
The expectile regression \citep{newey1987asymmetric} estimates the $\tau$th conditional expectile of $Y$ by minimizing an asymmetric quadratic loss function defined as
\begin{equation}\label{eq: expectile loss}
     \rho_{\tau}(x; u_\tau) =|\tau - I\{x \leq u_\tau\}|  (x - u_\tau)^2 ,
\end{equation}
where $\tau \in (0, 1)$ is the expectile level and $u_\tau$ denotes the $\tau$th expectile of error $\varepsilon$. which is the solution to the nonlinear equation \citep[Eq.~(2.7)]{newey1987asymmetric}.
A special case is $\tau  = 1/2$, where the expectile loss function reduces to the least-squares loss. 
Under our model, the conditional expectile of $Y$ at expectile level $\tau$ follows $X(\beta_0 + \gamma_0 u_\tau) + u_\tau$ \citep{newey1987asymmetric} and the coefficient vector $\beta_\tau = \beta_0 + \gamma_0 u_\tau$, where $u_\tau$ is the expectile of error $\varepsilon$ at $\tau$ in \eqref{eq:u_tau definition}. Since $\varepsilon$ has mean zero (Assumption~\ref{cond:error}), expectile regression at $\tau = 1/2$ does not identify the combination of $\gamma_0$ and $\beta_0$. Instead, $\beta_{1/2}$ only estimates $\beta_0$.

The proposed heteroscedasticity test is constructed based on the estimator of the parameter vector $\beta$ via minimizing the  empirical expectile loss function. Specifically, we consider the $\ell_1$-regularized expectile estimator defined as follows
\begin{eqnarray}\label{eq: expectile estimator n<p}
    \widehat\beta_\tau (\lambda) = \arg\min_{\beta \in \mathbbm R^p}\sum_{i = 1}^n \rho_{\tau} (Y_i -X_i\beta; u_\tau) + \lambda\|\beta\|_1,
\end{eqnarray}
where $\|\cdot\|_1$ denotes the $\ell_1$-norm and $\lambda$ denotes the tuning parameter. 
Due to technical restrictions of the AMP algorithm (see Section~\ref{sec: AMP algorithm}), we do not jointly optimize the intercept term $u_\tau$ and the slope vector $\beta$. Instead, we treat the term $u_\tau$ as a parameter of the expectile loss function. Recentering $Y$ by the intercept $u_\tau$ is important to guarantee the following AMP analysis when considering robust loss functions. In our simplest homoscedastic error case, plugging in $u_\tau$ reduces the conditional $\tau$-mean in \eqref{eq: conditional expectile of Y given X} to only $X\beta_\tau$. This intercept term is only ignorable when we address the conditional mean of $Y$ given X by the Assumption~\ref{cond:error}, the zero mean of $\varepsilon$. In homoscedastic settings, the error expectile $u_\tau$ can be estimated by solving  
\begin{equation}\label{eq: u_tau estimate}
    \hat u_\tau - \frac{1}{n}\sum_{i = 1}^n r_i = [(2\tau - 1)/(1 - \tau)] \frac{1}{n}\sum_{i = 1}^n (r_i - \hat u_\tau)  I\{r _i \ge \hat u_\tau \},
\end{equation}
where $r_i = Y_i - \widehat\beta_\tau X_i$'s are obtained using an initial estimator $\widehat\beta$ such as the regularized least-squares estimators, the regularized Huber estimator, etc. However, estimating $u_\tau$ becomes challenging under heteroscedasticity since solving \eqref{eq: u_tau estimate} estimates the $\tau$'s expectile of $\sigma(X\gamma_0) \varepsilon$ instead of $\varepsilon$. We observe that $u_\tau$ is reflected in the intercept term of $\gamma_0u_\tau$ in the composition $\beta_\tau = \beta_0 + \gamma_0 u_\tau$. Since the COupled Sparse Asymmetric LEast Squares (COSALES) \citep{gu2016high} estimates $\beta_0$ and $\gamma_0 u_\tau$ simultaneously for any $\tau \in (0, 1), \tau \neq 1/2$, we use the intercept estimator of $\gamma_0 u_\tau$ from the COSALES  as a surrogate for the true $u_\tau$ in \eqref{eq:u_tau definition}. Further, to evaluate the estimation error of $u_\tau$ on the test, we compared the numerical results using the true and estimated $u_\tau$. 
We observe that using the estimated $u_\tau$ from the COSALES does not affect the test results. 
\tcr{We also investigated the performance of the test when plugging in \eqref{eq: u_tau estimate} under the null; see Table~5.} The numerical performance appears similar.


Importantly, when the regression errors are i.i.d., i.e., $\sigma(\cdot) = 1$, the slope $\beta$ does not vary across different expectile levels 
in the classical setting with fixed $p$ and $n\to \infty$.  We propose to test the violation of homoscedasticity under diverging $n, p$ based on the ground of common slope under homoscedasticity at different expectile levels $\tau_k, k = 1, \ldots, K$.

\section{Approximate message passing for expectile regression}\label{sec: AMP algorithm}
The Approximate Message Passing (AMP) analysis starts by regularizing the loss functions $\rho$ by using the proximal mapping operator defined as 
\begin{eqnarray}\label{eq: proximal loss function}
  \mbox{Prox}(z; b) = \arg \min_{x \in \mathbbm R} \{ b\rho(x) + \frac{1}{2}(x-z)^2 \},  
\end{eqnarray}
where the parameter $b >0$ is a positive scalar that controls the step size moving towards the minimum of the loss function $\rho$. 
In our case, the proximal operator for expectile loss \eqref{eq: expectile loss} can be obtained following the derivation in Section~\ref{sssec: proximal operator calculation} and can be expressed as follows
\begin{eqnarray}\label{eq:proximal expectile}
\mbox{Prox}_{\tau}(z;b) &=&
\left\{\begin{array}{ll}
   \big(z+ 2b (1-\tau)u_\tau \big)\big({2b(1-\tau) + 1}\big)^{-1} , & z\leq u_\tau  \nonumber\\
  \big({z+ 2b \tau u_\tau}\big)\big({2b\tau +1 }\big), & z > u_\tau
\end{array}\right. \\
&=& \frac{z - 2bu_\tau (\tau  - I\{z \leq u_\tau\})}{1- 2b(\tau - I\{z \leq u_\tau\})}.
\end{eqnarray}


With the analytical expressions of the proximal operator for the expectile loss function \eqref{eq:proximal expectile}, we further work out the expression of the effective score function \citep{donoho2016high} defined as
\begin{eqnarray}
\label{eq:Gtilde}
\widetilde G(z;b) = b  \partial\rho(x)|_{x = {\rm Prox}(z;b)}, \mbox{ with } b>0 \tcr{,}
\end{eqnarray}
where $\partial\rho$ is the subgradient of $\rho$ defined as
$\partial\rho(x) = \{y: \rho(u)\ge \rho(x) + y(u-x), \mbox{for all } u \}$.
This leads to the subgradient of the expectile loss function in \eqref{eq: derivative expectile loss} as
\begin{equation}\label{eq: derivative expectile loss}
    \partial\rho_{\tau}(z) = 2 |I\{z\leq u_\tau\} - \tau| (z- u_\tau).
\end{equation}
One could interpret the effective score function by associating it with the score function. By plugging \eqref{eq:proximal expectile} in \eqref{eq: derivative expectile loss} and using the definition of the effective score function in \eqref{eq:Gtilde}, we obtain the effective score function as follows 
\begin{eqnarray}\label{eq:effective score expectile}
\widetilde G_{\tau}(z;b) = 
\left\{\begin{array}{ll}
   \frac{2b(1-\tau)}{2b(1-\tau) + 1 } (z - u_\tau) , & z \leq u_\tau \\
  \frac{2b\tau}{2b\tau + 1} (z- u_\tau) , & z >u_\tau
\end{array}\right. 
\end{eqnarray}
%
\tcr{To analyze regularized estimators under a sparse structure, specifically, when $ s/p \to \omega$, where $s$ denotes the number of nonzero components of $\beta$, \citet{bradic2016robustness,zhou2020detangling} incorporated sparsity by introducing the rescaled effective score function}
\begin{equation}\label{eq:scaled eff score}
    G(z; b) = \delta \omega^{-1} \widetilde{G}(z; b).
\end{equation}

\subsection{The approximate message passing algorithm}\label{ssec: AMP algorithm}
We consider the Approximate Message Passing (AMP) algorithm to approximate the estimators in 
\eqref{eq: expectile estimator n<p} in the sense that the averaged difference between the estimators obtained from the AMP algorithm and the estimators by minimizing the loss functions converges to zero almost surely, see \citet[p.~2562, Eq.~(23) and below]{zhou2020detangling} for a summary of discussions in the relevant references \citep{bayati2011lasso,donoho2016high,sur2019likelihood,bradic2016robustness,huang2020asymptotic}.

The AMP algorithm is an iterative algorithm indexed by $t$ consisting of three steps that belong to the general recursion in \citet{bayati2011dynamics}. We will describe the general recursion when stating the theory in Section~\ref{ssec:state evolution for multiple estimator}. Let the matrix $\bm X$ satisfy Assumption~\ref{cond:design}.
\begin{enumerate}[label = (A4)]
\item \label{cond:design}
     A standard Gaussian design: $X_i, i = 1, \ldots, n$ are i.i.d. copies of $X$, where $X^\top \in \mathbbm R^p$ with $X_{j}\sim N(0,1/n), j=1,\ldots,p$ independent and identically distributed components. 
\end{enumerate}

Since $\beta_\tau$ is a sparse vector with most components being zero, the algorithm starts with an initialization $\widehat\beta_{(0)} = 0, t = 0$ and proceed at $t \ge 1$ following steps:
\begin{enumerate}
    \item[] \textbf{Step 1. Adjust residuals} \\
    The adjusted residuals $z_{(t)}$ are obtained by
     \begin{eqnarray}\label{eq:adjust_resid}
    \lefteqn{z_{(t)} =  {\bm Y} - {\bm X}\widehat\beta_{(t)} + }
      \\& n^{-1} G(z_{(t-1)}; b_{(t-1)})  \nonumber
      \displaystyle\sum_{j=1}^p I\left\{ \eta\big(\widehat\beta_{(t-1),j} + {X}_{j} G(z_{(t-1)}; b_{(t-1)}); \theta_{(t-1)}\big) \not= 0\right\},
    \end{eqnarray}
    where $G$ is defined in \eqref{eq:scaled eff score} and $\eta(\cdot; \theta_{(t-1)})$ is the soft-thresholding operator with parameter $\theta_{(t-1)}$ defined as $\eta(\beta; \theta) = \mbox{sgn}(\beta)\max(|\beta| - \theta,0)$. 
    \item[] \textbf{Step 2. Update the parameter for effective score function}\\
    The scalar parameter $b_{(t)}$ is chosen such that the average of the effective score function has slope 1, i.e.,
    $
    \frac{1}{n} \sum_{i = 1}^n \partial_1 G (z_{i, (t)}; b_{(t)}) = 1
    $, where $\partial_1 G(z; b)$ denotes the derivative of $G$ w.r.t. the first argument $z$ in case of differentiable function $G$.
    Hence, for the expectile regression with calculations in Section~\ref{sssec: update b calculation}, the parameter $b_{(t)}$ is updated by solving
\begin{equation}\label{eq: update b}
    \frac{\omega}{\delta} =  \frac{2b(1-\tau)}{2b(1-\tau) + 1 } \frac{\#\{i: z_{(t), i} \leq u_\tau\}}{n} +  \frac{2b\tau}{2b\tau + 1 } \frac{\#\{i: z_{(t), i} > u_\tau\}}{n},
\end{equation}
where $\omega$ is estimated by the fraction of the nonnull components in an initial estimator $\widehat\beta$, $u_\tau$ is replaced by plugging in an estimator from the COSALES algorithm, see discussions under \eqref{eq: u_tau estimate}. The update $b_{(t)}$ at step $t$ is taken to be the first value on a fine grid that solves \eqref{eq: update b}.

   \item[] \textbf{Step 3. Update the estimators for $\beta$}\\
 Update the coefficient estimator for $\beta$ as follows
    \begin{eqnarray}\label{eq:pseudo_data and eq:esti}
     \widehat\beta_{(t+1)} = \eta (\widetilde\beta_{(t)}; \theta_{(t)}), \mbox{ where }   \widetilde\beta_{(t)} = \widehat\beta_{(t)} + {\bm X}^\top G(z_{(t)}; b_{(t)}).
    \end{eqnarray}
The estimator $\widehat\beta$ is an approximation for regularized estimators in \eqref{eq: expectile estimator n<p}, where the soft-thresholding operator $\eta(\cdot; \theta)$ is the proximal operator for the regularizer $\|\beta\|_1$. The tuning parameter $\theta$ is chosen to minimize the mean squared error (MSE) of $\widehat\beta$. As suggested in \citet{donoho1994minimax, donoho1998minimax, donoho2009message,bayati2011lasso}, we take $\theta_{(t)}  = \alpha \bar\zeta_{{\rm emp}(t)}$ for scalar $\alpha$ and 
\begin{equation}\label{eq: zeta est}
    \bar\zeta^2_{{\rm emp}, (t)} = \frac{1}{n} \sum_{i = 1}^n G(z_{i, (t)}; b_{(t)})^2.
\end{equation}
In practice, one chooses a grid of values for $\alpha$, then obtains the corresponding estimator $\widehat\beta$ and the estimated MSE for each value of $\alpha$ on the grid. The optimal $\alpha$ is chosen to be the one with the minimum estimated MSE. The details can be found in \citet[Algorithm~1]{zhou2020detangling}. We also implement practical stopping rules: the maximum number of iterations is set to \( T = 100 \), and convergence is declared when  
$\mbox{MSE}(\widehat{\beta}_{(t+1)}, \widehat{\beta}_{(t)}) = \frac{1}{p} \sum_{j=1}^p (\widehat{\beta}_{j,(t+1)} - \widehat{\beta}_{j,(t)})^2 < \mbox{tol}$.
In the simulations, the tolerance level $\mbox{tol}$ is set to $10^{-8}$ for high sparsity settings and $10^{-7}$ for medium sparsity settings.

Details of the population version of the parameter $\bar\zeta_{{\rm emp}, (t)}^2$ can be found in Section~\ref{ssec:state evo}.
The estimator $\widetilde\beta$ is of major interest in this project since it is connected to the debiased estimator \citep{javanmard2014hypothesis,van2014asymptotically} and follows a Gaussian distribution asymptotically \citep{bayati2011lasso, bayati2013estimating}. We will discuss this in Section~\ref{ssec:state evolution for multiple estimator}. 
\end{enumerate}

\subsection{State evolution}\label{ssec:state evo}
When $p, n \to \infty$, the averaged-over-components performance of estimators $\widetilde\beta_{(t)}$ and $\widehat\beta_{(t)}$ from the AMP algorithm can be described using two important parameters $\bar\zeta_{(t)}^2$, $\bar\sigma_{(t)}^2$. Assume that the following limit exists $\bar\sigma_{(0)}^2 = \lim_{p\to \infty}(p \delta)^{-1} \|q_{(0)}\|^2 < \infty$.
With such initialization, the state evolution recursion for $t\in \{0\} \bigcup \mathbbm{N}$ is defined as 
\begin{equation}\label{eq:state_evolution_zeta}
      \bar\zeta_{(t)}^2 = \E [G(\varepsilon  + \bar\sigma_{(t)}Z; b_{(t)})^2],
  \end{equation}
\begin{equation}\label{eq:state_evolution_sigma}
    \bar\sigma_{(t)}^2 = \delta^{-1} \E[ (\eta(B + \bar\zeta_{(t-1)}Z; \theta_{(t - 1)} ) -  B )^2],
  \end{equation}
   where $B$ denotes the limiting form of the components of $\beta$ in \eqref{eq: hetero linear model} and $Z \sim N(0, 1)$ is a standard Gaussian random variable that is independent of $B$ and $\varepsilon$. The parameter $\bar\zeta_{(t)}$ is the limiting version of \eqref{eq: zeta est} when $n, p \to \infty$ \citep[Section~B.2.3]{zhou2020detangling}. 
   
   The readers might wonder if \eqref{eq:state_evolution_zeta} and \eqref{eq:state_evolution_sigma} still hold under heteroscedasticity since the relevant literature has not discussed this issue and often assumes homoscedasticity. It is worth mentioning that the state evolution recursion in \eqref{eq:state_evolution_zeta} and \eqref{eq:state_evolution_sigma} is valid as long as the empirical distribution of $\varepsilon$ and the components of $\beta$ converge weakly \citep[Theorem~2]{bayati2011dynamics}. We will discuss this in detail in Section~\ref{ssec:state evolution for multiple estimator}.
 
  Using the two parameters, we can analyze the limiting performance of the estimator $\widehat\beta$ from the AMP algorithm.
  A typical example is analyzing the asymptotic MSE (AMSE), defined as 
   $\AMSE(\widehat\beta_{(t+1)}, \beta) = \lim_{p\to\infty} \frac{1}{p}\sum_{j = 1}^p(\widehat\beta_{(t+1),j} - \beta_j)^2$.
It is easy to show that AMSE complies with 
\begin{eqnarray*}
\AMSE(\widehat\beta_{(t+1)}, \beta) =
\lim_{p \to \infty} \frac{1}{p} \sum_{j = 1}^p \Big(\eta(\widetilde\beta_{(t),j} ; \theta_{(t)}) - \beta_{j}\Big)^2 
\stackrel{a.s.}{=} \E[\{\eta(B + \bar\zeta_{(t)}Z; \theta_{(t}) ) -  B \}^2],
\end{eqnarray*}
which corresponds to \citet[Eq.~(3.4)]{bradic2016robustness} with $\widetilde\beta_{(t),j}$ in (\ref{eq:pseudo_data and eq:esti}). 
An important observation from \eqref{eq:state_evolution_sigma} is that, at each iteration step $t$, the estimator $\widetilde\beta_{(t),j}$ is obtained by adding a Gaussian noise $N(0, \bar\zeta^2_{(t)})$ to the $j$th component of the true parameter vector $\beta_{ j}$. This suggests that $\widetilde\beta_{(t)} \sim N(\beta, \bar\zeta^2_{(t)}I_{p\times p})$ for finite marginals. Similar results are also obtained for the Lasso estimator in \citet{bayati2011lasso}, which is a special case of the $\ell_1$-regularized expectile regression when $\tau = 1/2$. We will come back to the asymptotic normality of $\widetilde\beta_{(t)}$ in Section~\ref{ssec:state evolution for multiple estimator} after introducing the main technical Lemma for multiple estimators $\widetilde\beta_{k, (t)}, k= 1,1 \ldots, K$.

\subsection{Asymptotic analysis for multiple estimators}\label{ssec:state evolution for multiple estimator}
The heteroscedasticity test for expectile regression is based on the difference of the estimators at different expectile levels $\tau_k, k = 1, \ldots, K$. The test statistic is based on asymptotic analysis for multiple estimators $\widehat\beta_{k}$'s among which the correlation is nontrivial when estimated using the same dataset \citep[Lemma~1]{zhou2020detangling}. Since estimating the correlation requires information on the unknown parameter vector $\beta$, we further extend the discussion on the correlations to the adjusted residuals $z_{k, (t)}$'s in the main technical result---Lemma~\ref{lemma:converge_cov} and construct an adjusted-residual-based estimator for the covariance between $\widetilde\beta_{k}$'s. Further, we want to address that the technical results, including the general recursion and Lemma~\ref{lemma:converge_cov} in this section, hold for similar loss functions, such as the quantile loss function. 

We describe the general recursions in \citet{bayati2011dynamics} for the completeness of the paper and a better connection of Lemma~\ref{lemma:converge_cov} and the AMP algorithm in Section~\ref{ssec: AMP algorithm}. Under the heteroscedastic model \eqref{eq: hetero linear model}, the coefficient vector $\beta_k = \beta_0 + u_k \gamma_0$, where $u_k$ is a constant featuring the strength of heteroscedasticity which depends on $\varepsilon$. For expectile regression, $u_k$ is the error expectile. By Assumptions~\ref{cond:gamma} and \ref{cond:beta}, the empirical distribution of the parameter vector $\beta_k$ converges weakly to $B_k = B_0 + u_k \Gamma_0$. In cases where $u_k = 0$ or homoscedasticity such that $\Gamma_0 = 0$ degenerates, $B_k = B_0$.

In addition, let a noise vector $\bm{\varepsilon} \in \mathbbm R^n$ and the coefficient vector $\beta_k \in \mathbbm R^p$ satisfy Assumptions~\ref{cond:error}, \ref{cond:beta}. By fixing the initial condition $q_{k, (0)}$ and recall the convention $m_{k,(-1)} = 0$, the general recursion for $t\ge 0$ is defined as
\begin{eqnarray}\label{eq: general recursion}
    &&h_{k, (t+1)} = {\bm X}^\top m_{k, (t)} - \xi_{1,k, (t)}q_{k, (t)},\ \ \ m_{k, (t)} = g_{1,k, (t)}(d_{k, (t)}, \bm\varepsilon)  \nonumber\\
   &&d_{k, (t)} = {\bm X} q_{k, (t)} - \xi_{2,k,(t)} m_{k, (t-1)},\ \ \  q_{k,(t)} = g_{2,k,(t)}(h_{k,(t)}, \beta_k) 
\end{eqnarray}
where 
\begin{equation}\label{eq: effective score general}
    \xi_{1,k, (t)} = n^{-1} \sum_{i = 1}^n \partial_1 g_{1,k,(t)}(d_{k,(t), i}, \varepsilon_i),
\end{equation}
\begin{equation}\label{eq: mse general}
    \xi_{2, (t)} = (\delta p)^{-1}\sum_{j = 1}^p\partial_1 g_{2,k,(t)}(h_{k,(t), j}, \beta_{k,j}).
\end{equation}
As a special case, the AMP algorithm takes
\begin{eqnarray*}
    &h_{k, (t+1)} = \beta_{k} - {\bm X}^\top G_k(z_{k, (t)}; b_{k, (t)}) - \widehat\beta_{k, (t)}, \\
    &q_{k, (t)} = \widehat\beta_{k, (t)} - \beta_k, \qquad z_{k,(t)} = \bm\varepsilon - d_{k, (t)},
\end{eqnarray*}
which indicates 
\begin{equation}\label{eq: general recursion for AMP--residual}
    d_{k,(t)} = \bm\varepsilon - z_{k,(t)}, \qquad m_{k,(t)} = - {G_k}(z_{k, (t)}; b_{k, (t)})
\end{equation}
with the functions $g_{1,k,(t)}(x_1,x_2) = - G_k(x_2-x_1; b_{k,(t)})$, and $g_{2, k, (t)}(x_1, \beta) = \eta(\beta - x_1; \theta_{k, (t)}) - \beta$. Notice $\widehat\beta_{k, (t)}$ the regularized estimator in \eqref{eq:pseudo_data and eq:esti}, and 
\begin{equation}\label{eq: h_k}
    h_{k, (t+1)} = \beta_k - \widetilde\beta_{k, (t)}.
\end{equation}
This detail is mathematically trivial but allows us to apply Lemma~\ref{lemma:converge_cov} directly to obtain the asymptotic normality of any pairs $(\widetilde\beta_{k_1, (t)}, \widetilde\beta_{k_2,(t)})$, as well as later obtain the power function for the proposed hypothesis test in Theorem~\ref{thm:power function of test}. 

Before presenting the main technical Lemma, we assume a moment condition Assumption~\ref{cond:lemma_convergence}, which also complements Assumptions~\ref{cond:error}, \ref{cond:gamma}, and \ref{cond:beta}. Recall in Assumptions~\ref{cond:error}, \ref{cond:gamma}, and \ref{cond:beta}, we state that the empirical measures of the components of the vectors $\beta_0$, $\gamma_0$, and $\bm \varepsilon$ converge weakly to the probability measures of random variables $B_0$, $\Gamma_0$, and $\varepsilon$, respectively. Let $\widehat{f}_{\beta_0(p)}$, $\widehat{f}_{\gamma_0(p)}$, $\widehat{f}_{\bm\varepsilon(p)}$ denote the empirical measures of $\beta_0$, $\gamma_0$, $\bm\varepsilon$, i.e., putting equal weights on the components of the vectors. Specifically, for the sequence of vectors $\bm \varepsilon(p)$, we index the sequence by $p$ to unify the notation following \citet{bayati2011dynamics}. This sequence index is legit since our asymptotic theory is derived under $n/p \to \delta \in (0, 1]$. Then, Assumption~\ref{cond:lemma_convergence} can be stated as follows
\begin{enumerate}[label = (A5)]
 \item \label{cond:lemma_convergence}
 For some $\kappa >1 $, let the empirical measures $\widehat{f}_{\beta_0(p)}$, $\widehat{f}_{\gamma_0(p)}$, and $\widehat{f}_{\bm\varepsilon(p)}$ converge weakly to the probability measures $f_{B_0}$, $f_{\Gamma_0}$, and $f_\varepsilon$. 
 In addition, assume the following $(2\kappa - 2)$th moment conditions hold
 \begin{enumerate}
         \item $\lim_{p \to \infty} \E_{\widehat{f}_{\beta_0(p)}}(B_0^{2\kappa - 2}) = \E_{f_{B_0}}(B_0^{2\kappa - 2}) < \infty$
         \item $\lim_{p \to \infty} \E_{\widehat{f}_{\gamma_0(p)}}(\Gamma_0^{2\kappa - 2}) = \E_{f_{\Gamma_0}}(\Gamma_0^{2\kappa - 2}) < \infty$
         \item $\lim_{p \to \infty} \E_{\widehat{f}_{\bm\varepsilon (p)}}(\varepsilon^{2\kappa - 2}) = \E_{f_{\varepsilon}}(\varepsilon^{2\kappa - 2}) < \infty$
         \item $\lim_{p \to \infty} \E_{\widehat{f}_{q_{k}(p)}}(B_k^{2\kappa - 2}) < \infty$.
     \end{enumerate} 
\end{enumerate}
The constant $\kappa$ depends on the pseudo-Lipschitz order of function $g_{1, k, (t)}$, see Definition~\ref{def:pseudo_lipschitz} for the definition of pseudo-Lipschitz functions and Lemma~\ref{lemma:converge_cov} for associating the order $\kappa$ with the function $g_1 = \tilde\psi_{c'}, \tilde\psi_{c''}$. As described above, in the special case where the general recursion \eqref{eq: general recursion} is the AMP algorithm, function $g_1$ takes a more specific form, that is, the negative effective score function $-G_k(\cdot; \cdot)$.

\begin{lemma}
\label{lemma:converge_cov}
Let the sequences of design matrices $\{\bm X(p)\}$, location coefficient vectors $\{\beta_0(p)\}$, heteroscedasticity parameter vectors $\{\gamma_0(p)\}$, error vectors $\{\bm{\varepsilon}(p)\}$, initial condition vectors $\{q_k(p)\}$ be weakly convergent sequences for $K$ recursions satisfying Assumptions \ref{cond:error}--\ref{cond:lemma_convergence}. Let $\{\bar\sigma_{k, (t)}^2, \bar\zeta_{k, (t)}^2\}$ be defined uniquely by the recursions in \eqref{eq:state_evolution_zeta} and \eqref{eq:state_evolution_sigma}. These are the state evolution parameters for the $k$th estimation with initialization $\bar\sigma^2_{k, (0)} = \lim_{n\to \infty}(p\delta)^{-1} \|q_{k,(0)}\|^2$. Further, we assume a constant $u_k$ for each of the $K$ recursions specifying the signal strength of the heteroscedasticity parameter, such that the parameter vector \eqref{eq: hetero linear model} has the composition 
$\beta_k = \beta_0 + u_k \gamma_0$.
Assume Lemma 1 in \citet{bayati2011dynamics} holds individually for each of the $K$ recursions. Additionally, for all pseudo-Lipschitz functions $\tilde\psi_{\rm c}: \mathbbm R^{t+2} \to \mathbbm R$ of order $\kappa_{\rm c}$ for some $1 \leq \kappa_{\rm c} = \kappa/2$ with $\kappa$ as in \ref{cond:lemma_convergence} and the iteration index $t \geq 0, t\in \mathbbm Z$,
\begin{eqnarray}\label{eq:correlation of betas convergence}
{\lim_{p\to \infty} \frac{1}{p}\sum_{j = 1}^p \tilde\psi_{\rm c}(h_{k_1, (1), j},\ldots, h_{k_1, (t+1), j}, \beta_{k_1, j})
 \tilde\psi_{\rm c}(h_{k_2, (1), j}, \ldots, h_{k_2, (t+1), j}, \beta_{k_2, j})} \stackrel{a.s.}{=} \nonumber\\
 \E[\tilde\psi_{\rm c}(\bar\zeta_{k_1, (0)}Z_{k_1, (0)},\ldots, \bar\zeta_{k_1, (t)}Z_{k_1, (t)}, B_{k_1})
 \tilde\psi_{\rm c}(\bar\zeta_{k_2, (0)}Z_{k_2, (0)},\ldots, \bar\zeta_{k_2, (t)}Z_{k_2, (t)}, B_{k_2})]. \nonumber\\
\end{eqnarray}
And, for all pseudo-Lipschitz functions $\tilde\psi_{\rm c'},\tilde\psi_{\rm c''} : \mathbbm R^{t+2} \to \mathbbm R$ of order $\kappa_{\rm c'}, \kappa_{\rm c''}$ for some $1 \leq \kappa_{\rm c'}, \kappa_{\rm c''} \leq \kappa/2$, where $\kappa_c' + \kappa_c'' = \kappa$ with $\kappa$ as in \ref{cond:lemma_convergence},
\begin{eqnarray}\label{eq:correlations of error convergence}
{\lim_{n\to \infty} \frac{1}{n}\sum_{i = 1}^n \tilde\psi_{\rm c'}(d_{k_1, (0), i},\ldots, d_{k_1,(t), i}, \varepsilon_i)
 \tilde\psi_{\rm c''}(d_{k_2, (0), i}, \ldots, d_{k_2, (t), i}, \varepsilon_i)} \stackrel{a.s.}{=} \nonumber\\
 \E[\tilde\psi_{\rm c'}(\bar\sigma_{k_1, (0)}\hat Z_{k_1, (0)},\ldots, \bar\sigma_{k_1, (t)}\hat Z_{k_1, (t)}, \varepsilon)
 \tilde\psi_{\rm c''}(\bar\sigma_{k_2, (0)}\hat Z_{k_2, (0)},\ldots, \bar\sigma_{k_2, (t)}\hat Z_{k_2, (t)}, \varepsilon)]. \nonumber\\
\end{eqnarray}
where $(Z_{k, (0)}, \ldots, Z_{k, (t)}), (\hat Z_{k, (0)}, \ldots, \hat Z_{k, (t)})$ are multivariate Gaussian with, $k = k_1, k_2$, is a $(t+1)$-dimensional zero-mean multivariate standard normal vector independent of $B_k$, $\varepsilon$; at iteration $t$, $(Z_{k_1, (t)}, Z_{k_2, (t)})$ and $(\hat Z_{k_1, (t)}, \hat Z_{k_2, (t)})$ are standard normal 2-vectors with covariance not necessarily equal to zero.

Further, for $0<t_1, t_2<t$, with $\stackrel{P}{\to}$ denoting convergence in probability, it holds that 
\begin{eqnarray}\label{eq: lemma 1c relevant}
    \frac{1}{n} \sum_{i = 1}^n m_{k_1, (t_1), i}  m_{k_2, (t_2), i} \stackrel{P}{\to}\lim_{p\to \infty} \frac{1}{p}\sum_{j = 1}^p h_{k_1, (t_1+1), j}  h_{k_2, (t_2+1), j}.
\end{eqnarray}
\end{lemma}

Lemma~\ref{lemma:converge_cov} generalizes \citet[Lemma~1(b)]{bayati2011dynamics} by considering the correlation structure between multiple recursions. Further, Lemma~\ref{lemma:converge_cov} extends the analysis to account for error heteroscedasticity, which has not been addressed in the relevant literature. When $k_1 = k_2$ and the heteroscedastic parameter vector $\gamma_0 = \bm 0$ is a null vector, Lemma~\ref{lemma:converge_cov} reduces to \citet[Lemma~1(b)]{bayati2011dynamics}. Further, the covariance matrices of the standard normal 2-vectors $(Z_{k_1, (t)}, Z_{k_2, (t)})$ and $(\hat Z_{k_1, (t)}, \hat Z_{k_2, (t)})$ will be involved in constructing the test statistic in Section~\ref{sec: test}. 
\begin{remark}
We thank a reviewer for drawing our attention to Theorem~1 and Lemma~2 in \citet{javanmard2013state}. Their result agrees with our conclusion in \eqref{eq:correlation of betas convergence}, which shows the asymptotic approximation of the covariance of the structure of multiple estimators. Our additional contributions lie in establishing \eqref{eq:correlations of error convergence} and \eqref{eq: lemma 1c relevant}, which are of independent interest to statistics and econometrics, as they are important for establishing Theorem~\ref{thm: covariance construction}, \eqref{eq: cov converge}, and demonstrate that the correlation between estimators is linked to the correlation of residuals.
\end{remark}

To guarantee an estimable test statistic, we propose an adjusted-residual-based estimator in Theorem~\ref{thm: covariance construction}, which is directly obtainable from the AMP algorithm. The construction of the covariance matrix estimators is included in Section~\ref{sssec: covariance construction proof} by directly applying \eqref{eq: lemma 1c relevant}. It is worth noting that \eqref{eq: lemma 1c relevant} is a contribution to the existing AMP literature. The significance of \eqref{eq: lemma 1c relevant} is that it shows the correlation between estimators (noticing $h_{k,(t+1)} = \beta - \widetilde\beta_{k, (t)}$ in \eqref{eq: h_k}) from two sequences is asymptotically equivalent to the product of two effective score functions. This avoids the problem of lacking information about the true parameter vector $\beta$. Additionally, Theorem~\ref{thm: covariance construction} validates that $Z_{k_1}$ and $Z_{k_2}$ are correlated by \eqref{eq: cov exist}.

\begin{theorem}\label{thm: covariance construction}
    Under the assumptions in Lemma~\ref{lemma:converge_cov} {and $\kappa = 2$} for any two recursions $k_1, k_2 = 1, \ldots, K$, 
    with probability 1, it holds that 
    \begin{equation}\label{eq: cov exist}
      \lim_{p\to \infty}  \frac{1}{p}\sum_{j = 1}^p \big(\widetilde\beta_{k_1, (t),j} - \beta_{k_1, j}\big)\big(\widetilde\beta_{k_2, (t), j} - \beta_{k_2, j}\big) = \bar\zeta_{k_1, (t)} \bar\zeta_{k_2, (t)}\Cov(Z_{k_1}  ,Z_{k_2}),
    \end{equation}
    \begin{eqnarray}\label{eq: cov converge}
        \lefteqn{\lim_{n\to\infty}\frac{1}{n} \sum_{i=1}^n G_{k_1}(z_{k_1, (t), i}; b_{k_1, (t)})  G_{k_2}(z_{k_2, (t), i}; b_{k_2, (t)})} &\nonumber
         \\ 
         & \qquad\qquad= \E[G_{k_1}(\varepsilon + \bar\sigma_{k_1}\hat Z_{k_1}; b_{k_1, (t)}) G_{k_2}(\varepsilon + \bar\sigma_{k_2}\hat Z_{k_2}; b_{k_2, (t)})].
    \end{eqnarray}
The estimator constructed as follows
\begin{equation}\label{eq: cov est}
    \widehat\zeta_{k_1, k_2, (t)} = \frac{1}{n}\sum_{i=1}^n G_{k_1}(z_{k_1,(t), i}; b_{k_1, (t)})  G_{k_2}(z_{k_2, (t), i}; b_{k_2, (t)})
\end{equation}
is a consistent estimator of the covariance $\bar\zeta_{k_1, (t)}\bar\zeta_{k_2, (t)} \Cov(Z_{k_1, (t)}, Z_{k_2, (t)})$.
\end{theorem}

Theorem~\ref{thm: covariance construction} first validates our claim in Lemma~\ref{lemma:converge_cov} that $Z_{k_1, (t)}, Z_{k_2, (t)}$ are correlated. The correlation between $Z_{k_1, (t)}, Z_{k_2, (t)}$ describes the averaged correlation between the components of $\widetilde\beta_{k_1, (t)}$ and $\widetilde\beta_{k_2, (t)}$. A similar conclusion holds for $\hat Z_{k_1, (t)}, \hat Z_{k_2, (t)}$. Further, the covariance between $Z_{k_1, (t)}, Z_{k_2, (t)}$ can be approximated by using the adjusted residuals without any information about $\beta$. 

The last piece of information before moving to the test statistic is the joint asymptotic distribution of the pair $(\widetilde\beta_{k_1}, \widetilde\beta_{k_2})$.
\begin{lemma}[Joint normality]\label{lemma:normality}
 Under the assumptions in Lemma~\ref{lemma:converge_cov} for any two recursions $k_1, k_2 = 1, \ldots, K$ and for any pseudo-Lipschitz function $\psi: \mathbbm R \to \mathbbm R$ of order $\kappa_c$ in \eqref{eq:correlation of betas convergence}, the following convergence for estimators from the AMP algorithm holds with probability 1,
 \begin{equation}\label{eq: weak convergence of betatilde}
\lim_{p\to \infty}\frac{1}{p}\sum_{j = 1}^p
  \psi(\widetilde\beta_{k_1, (t), j})\psi(\widetilde\beta_{k_2, (t), j}) = 
 \E\big[\psi(B_{k_1}+ \bar\zeta_{k_1, (t)}Z_{k_1}) \psi(B_{k_2} + \bar\zeta_{k_2, (t)} Z_{k_2})\big],
 \end{equation}
 where $(Z_{k_1}, Z_{k_2})$ is a standard normal random 2-vector.  
\end{lemma}
The proof of Lemma~\ref{lemma:normality} is a direct application of Lemma~\ref{lemma:converge_cov}, see Section~\ref{sssec: bivariate normality convergence}. 
When writing in the matrix form, the conditional joint asymptotic normality 
     $ \bm{\tilde\xi} - \bm{\xi} \sim N\Big((u_{k_1}, u_{k_2}) \otimes \gamma_0 , \Sigma_{(t)} \otimes \bm I_{p\times p}\Big)$ holds for $k_1, k_2 = 1, \ldots, K$,
where 
$\bm{\tilde \xi} = (\widetilde\beta_{k_1, (t)},\widetilde\beta_{k_2, (t)})$, $\bm\xi = (1,1)\otimes \beta_0$, $\bm I_{p\times p}$ is a $p\times p$ identity matrix. The covariance matrix $\Sigma_{(t)}$ has $(k_1, k_2)$th component $\Sigma_{k_1, k_2, (t)} = \bar\zeta_{k_1, k_2, (t)}$, which can be estimated by \eqref{eq: cov est}.

Under homoscedasticity where $\gamma_0 = \bm 0$, $\widetilde\beta_{k_1}$ and $\widetilde\beta_{k_2}$ have a common center $\beta_0$. In addition, variance components $\bar\zeta_{k, (t)}^2, k = 1, \ldots, K$ are inflated under heteroscedasticity. This phenomenon can be observed by the inflation of the standard deviation $\bar\sigma_{k, (t)}^2$ in \eqref{eq:state_evolution_zeta}, which is caused by the extra heteroscedastic parameter in \eqref{eq:state_evolution_sigma}.

\section{Heteroscedasticity test}\label{sec: test}
Although the asymptotic normality of $\bm{\widetilde\xi}$ holds for all iterations $t = 0, 1, \ldots$ of the AMP algorithm, we use the estimators $\bm{\widetilde\xi}$ at $t \to \infty$ to construct test statistics. When the iteration $t \to \infty$, the $l_2$-norm of the difference of the AMP output and the estimator in \eqref{eq: expectile estimator n<p} converges to zero almost surely, see \citet[Theorem~2]{bradic2016robustness}. In practice, we consider a stopping rule and define the algorithm's convergence in Section~\ref{ssec: AMP algorithm} under \eqref{eq: zeta est}. At the stop, we obtain an estimator approximating the estimator defined in \eqref{eq: expectile estimator n<p}. Hence, the iteration index $t$ is dropped in the following sections.
We consider testing the components of the vector $\gamma_0$ 
\begin{equation}\label{eq:hypothesis}
    H_{0, j}: \gamma_{0, j} = 0 \mbox{ versus } H_{1, j}: \gamma_{0, j} \neq 0, j = 1, \ldots, p.
\end{equation}
Lemma~\ref{lemma:normality} states that the $K$ expectile estimators follow a multivariate Gaussian distribution asymptotically with a common center $\beta_0$ under homoscedasticity. And the rows of $\bm{\widetilde\xi}$ are i.i.d. following a joint Gaussian distribution with covariance matrix $\Sigma$.

For $\gamma_{0, j}$, the test statistic is defined as
\begin{equation}\label{eq: test stat}
    T_j = \frac{\Delta \bm{\tilde\xi}_j^\top }{\sqrt{\Delta \Sigma \Delta^\top}}, j = 1, \ldots, p.
\end{equation}
Following \citet[Section~4.1]{newey1987asymmetric}, we consider the contrast between two expectile levels for the testing problem, i.e., $\bm{\tilde\xi} = (\widetilde\beta_{\tau_1}, \widetilde\beta_{\tau_2})$, $\Delta = (1, -1)$, and the components of the covariance matrix $\Sigma$ are replaced by their estimators in \eqref{eq: cov est}. Then, by Lemma~\ref{lemma:normality}, the test statistic $T_j$ follows a standard normal distribution $N(0, 1)$. For the two-sided alternative in \eqref{eq:hypothesis}, the $p$-value can be calculated by
\begin{equation}\label{eq:pvalue H_0}
    P_{j} = 2\{1 - \Phi(|\frac{\Delta \bm{\tilde\xi}_j^\top }{\sqrt{\Delta \Sigma \Delta^\top}} |) \},
\end{equation}
with rejection region $R = (-\infty, \Phi^{-1}(\alpha/2)) \cup (\Phi^{-1}(1 - \alpha/2), \infty)$, where $\Phi$ denotes the cumulative distribution function of a standard normal distribution $N(0,1)$. The null hypothesis is rejected if $P_j \le \alpha$ at nominal level $\alpha$.
To analyze the power function on the population scale when $n, p \to \infty$, we consider for any $j$th component 
    $\vartheta(\gamma_{0, j}) = I\big( T_j \in R, \Gamma_0 = \gamma_{0, j}\big)$.
Then, the limiting expression of the function $\vartheta$ is obtained in Theorem~\ref{thm:power function of test} with proof in Appendix~\ref{appendix: proof of convergence of power function}.

\begin{theorem}\label{thm:power function of test}
    The function denoted by $\vartheta(\gamma_{0, j}), j = 1, \ldots, p$ complies with 
    \begin{eqnarray}\label{eq:power function cal}
       &&\lim_{p\to\infty}\frac{1}{p} \sum_{j = 1}^p \vartheta(\gamma_{0, j}) \stackrel{a.s.}{=} \\
&&\E_{\Gamma_0}\Big[1 - \Phi\big(\Phi^{-1}(1 - \alpha/2) - \Gamma_0\frac{\Delta (u_{\tau_1}, u_{\tau_2})^\top}{\sqrt{\Delta \Sigma \Delta^\top}}\big)  + 
  \Phi\big(\Phi^{-1}(\alpha/2) - \Gamma_0\frac{\Delta (u_{\tau_1}, u_{\tau_2})^\top}{\sqrt{\Delta \Sigma \Delta^\top}}\big)
  \Big].\nonumber
    \end{eqnarray}
\end{theorem}

Theorem~\ref{thm:power function of test} can be used to confirm the size and power of the proposed test. Assume the nominal significance level of the test is $\alpha$. The size of the test $ P(T_j \in R | \Gamma_0 = 0)$ is approximated by 
\begin{eqnarray*}
   \lefteqn{\frac{\lim_{p\to \infty}\frac{1}{p}\sum_{j = 1}^p I(T_j \in R ,\Gamma_0 = 0)}{P(\Gamma_0 = 0)}}&&\\
   &\stackrel{a.s.}{=}&
   \frac{\E_{\Gamma_0}\Big[1 - \Phi\big(\Phi^{-1}(1 - \alpha/2) - 0\big)  + 
  \Phi\big(\Phi^{-1}(\alpha/2) - 0\big)
  \mid \Gamma_0 = 0\Big] P(\Gamma_0 = 0)}{P(\Gamma_0 = 0)}\\
  &=& \alpha,
\end{eqnarray*}
where the conditional expectation is obtained by \eqref{eq:power function cal} at $\Gamma_0 = 0$. This confirms that the size of a nominal $\alpha$ test is $\alpha$ asymptotically.

Further, the empirical size of the test is defined in \eqref{eq: empirical type I} and used to evaluate the performance of the proposed test. As presented in Tables~\ref{table: simulation test results tau0.2}, \ref{table: simulation test results tau0.6}, the empirical size, averaged over the components of $\gamma_0$, is approximately equal to the nominal level $\alpha$. 
Similarly, the power of the test is obtained by considering the nonnull components $\gamma_{0,j}$'s. However, from the right-hand-side of \eqref{eq:power function cal}, we notice that the power of the test depends on the averaged signal strength of $\Gamma_0$, the expectiles $(u_{\tau_1}, u_{\tau_2})$, and the covariance matrix $\Sigma$. The empirical power of the test defined in \eqref{eq: empirical power} is also used to evaluate the numerical performance of the test in Section~\ref{sec:numerical}.

\section{Decorrelation in AMP}\label{sec:decorrelation}
It is well known that the standard Gaussian design assumption, which is stated in Assumption~\ref{cond:design}, is important for the AMP framework. Researchers have made several attempts to relax this assumption.  One direction is to relax the Gaussian distribution assumption. For example, \citet{universality2015Bayati,chen2021universality} showed that the AMP framework could be extended to other random matrices; \citet{rangan2019vector,fan2022approximate} discussed the extension of AMP algorithms to right rotationally invariant design. \citet{li2023spectrum} showed the statistical applications of these algorithms.
However, the independence assumption is still essential in these works. A more relevant direction for this work is to relax the independence assumption in order to handle general design matrices with arbitrary covariance matrices. 
Consider a general design matrix in the sparse high-dimensional heteroscedastic linear model: 
$
  \tilde{\bm Y} =  \tilde{\bm X} \beta + \tilde{\bm\varepsilon},     
$
where $ \tilde{\bm Y} \in \mathbbm{R}^n$, $ \tilde{\bm X} \in \mathbbm{R}^{n\times p}$, $ \tilde{\bm\varepsilon} \in \mathbbm{R}^n$, and the rows of $\tilde{\bm X}$ are i.i.d. from a multivariate Gaussian distribution $N_p(0, {\tilde\Sigma})$. 
Several papers have discussed incorporating a general Gaussian design with strictly positive definite covariance matrices. 
For example, \citet[Definition~4.5, Corollary~4.6]{donoho2016high}, \citet[Section~3.1.1]{candes2020phase}, \citet[Section~2]{zhao2022asymptotic}, \citet{huang2022lasso} have proposed a method to decorrelate the design matrix using the nonsingular covariance matrix $\bm{\tilde{\Sigma}}$ by applying the transformation $(\bm{ \tilde X\tilde{\Sigma}}^{-1/2})/\sqrt{n}$. This decorrelation allows for satisfying Assumption~\ref{cond:design} and expands the applicability of AMP-based methods and theories. However, this decorrelation approach has limitations when dealing with sparse parameter vectors $\beta$ because the null components of $\beta$ become nonnull under the linear transformation $\sqrt{n}\bm{\tilde{\Sigma}}^{1/2}\beta$ and the corresponding back transformation. Additionally, in practical applications, the covariance matrix $\bm{\tilde{\Sigma}}$ is often unknown, and estimating it becomes challenging when $p>n$ unless certain structure assumptions are made. These challenges restrict the use of this approach in real applications.

In applications, we consider a different method for decorrelation in the AMP.  The validity of this type of decorrelation has been tested numerically in \citet[Table~3]{zhou2023automatic}. We utilize a ``generalized puffer transformation" on the data, resulting in transformed variables denoted as $\bm{{Y}} = F \tilde{\bm{Y}}$ and $\bm{{X}} = F \tilde{\bm{X}}$, where the transformed data satisfy the model $
    \bm{ Y}= \bm{ X}\beta + F { \tilde{\bm\varepsilon}}.
$
It is important to note that in our setting, $p>n$. To perform the transformation, we utilize matrices $\bm{U}\in\mathbbm{R}^{n\times n}$, $\bm{D}\in\mathbbm{R}^{n\times n}$, and ${\bm{V}}^\top\in\mathbbm{R}^{n\times p}$ obtained from the singular value decomposition $\bm{\tilde{X} = UDV^\top}$.
The ``generalized puffer transformation" was originally proposed in \citet{jia2015preconditioning}, where the authors aimed to decorrelate the design matrix to bypass the so-called ``irrepresentable condition" for the Lasso regression.  
The ``generalized puffer transformation" uses $F = \sqrt{p/n} \bm{U \hat D U}^\top$ where $\bm{ \hat D}$ is a diagonal matrix with diagonal elements $\bm{\hat D}_{ii} = I({\bm D}_{ii} \le 1/\sqrt{n})\sqrt{n}+I({\bm D}_{ii} > 1/\sqrt{n})/{\bm D}_{ii}, i=1,\ldots,n$. Subsequently, we apply the AMP algorithm and conduct the hypothesis test using the transformed data $(\bm{{Y}}, \bm{{X}})$.

Unlike the decorrelation step in \citet{candes2020phase,donoho2016high,zhao2022asymptotic}, our approach does not involve a transformation of $\beta$. Thus, it can preserve sparse structures. 
When $p> n$, we cannot guarantee that $\bm{X}$ precisely satisfies Assumption~\ref{cond:design}, but the correlation is largely reduced, as observed in the numerical study. 
Moreover, on the transformed data, the error term becomes $F \tilde{\bm\varepsilon}$.
Under the null hypothesis that the model \ref{eq: linear approximation of linear model} is homoscedastic, we have verified that $F\tilde{\bm{\varepsilon}}$ is independent of the transformed design matrix $\bm X = \sqrt{p/n}\bm{U V^\top}$. This observation suggests that the ``generalized puffer transformation" does not alter the null hypothesis. This explains why the proposed test retains its validity size, as demonstrated in our simulation study. For further details, please refer to Section~\ref{sec:6.2}.

\section{Numerical performance}\label{sec:numerical}
\subsection{Simulation}\label{ssec: simulation}

We consider simulation settings where $p = 500$, $n= 250$. The samples $X_i$ are generated from a standard Gaussian design, as specified in Assumption~\ref{cond:design}. The nonnull components of the true parameter vector $\beta_0$ are generated using a standard normal $N(0, 1)$ distribution and then multiplied by 5. Each setting is replicated $R = 400$ times. The parameter vector $\beta_0$ and design matrix $\bm X$ are fixed across $R$ replications using the same seed number, while the regression errors $\bm\varepsilon$ vary in each replication. 
We evaluate the test under six different error distributions: $N(0,1 )$, $t_3$, $\mbox{Lognormal}(0,1)$, $0.9N(-0.2, 0.25) + 0.1N(1.8, 0.01)$, $0.9N(0.2, 0.25) + 0.1N(-1.8, 0.01)$, and $0.95 N(0, 0.25) + 0.05 N(0, 4)$. The first three error distributions are randomly generated, centered, and scaled to have a standard deviation of 0.5. We aim to assess the test's performance using symmetric errors (represented by $N(0,1)$), heavy-tailed errors (represented by $t_3$), and fat-tailed errors (represented by $\mbox{Lognormal}(0,1)$). Additionally, we consider three mixture normal distributed errors. The first two, $0.9N(-2, 0.25) + 0.1N(18, 0.01)$ and $0.9 N(2, 0.25) + 0.1N(-18, 0.01)$, model clustered outliers on the right and left sides of the main error distribution, respectively. The last distribution---$0.95 N(0, 0.25) + 0.05 N(0, 4)$---adds additional variability to $N(0, 0.25)$.
For heteroscedasticity, we assume 5 components of $\gamma_0$ are nonnulls $(3, 1, -5, -5, -3)^\top$, and the rest components are all zeros. Similar to $\beta_0$, the components of $\gamma_0$ are fixed throughout $R$ replications using the same seed number. 


The individual test is evaluated by the averaged size (or false positive) proportion and power (or true positive proportion). Both proportions are obtained by averaging over simulation replications and the components of $\gamma_0$, which agrees with the construction of the power function in Theorem~\ref{thm:power function of test}. At significance level $\alpha$, the size can be estimated via 
\begin{equation}\label{eq: empirical type I}
\mbox{FP}(\alpha) =(p-s)^{-1}\sum_{j \in S^{\rm c}}{\sum_{r=1}^R I\{P_{r,j} \leq \alpha \}}/{R} \qquad\mbox{ (empirical size),}
\end{equation}
and the power can be estimated by
\begin{equation}\label{eq: empirical power}
\mbox{TP}(\alpha) =s^{-1}{\sum_{j \in S}{\sum_{r=1}^R I\{P_{r,j} \leq \alpha \}}/{R}} \qquad \mbox{(empirical power)}.
\end{equation}
For homoscedastic cases, the components of $\gamma_0$ are all zeros and $\|\gamma_0\|_0 = 0$; hence only FP values are calculated. 
For heteroscedasticity cases, five components of $\gamma_0$ are nonnull with magnitude $(3, 1, -5, -5, -3)^\top$. The signal strength of heteroscedasticity $u_\tau \gamma_0$ depends on the error expectile $u_\tau$. Since the errors are scaled to a standard deviation of 0.5, the heteroscedasticity signal strength is largely weakened by $u_\tau$.  

We consider two sparsity levels for fixed ratio $\delta = 0.5$, $p = 500$, $n = 250$.
\begin{enumerate}
    \item (high-sparsity) $\beta_0$ has $s= 5$ nonnull components.
    \item (medium-sparsity) $\beta_0$ has $s = 40$ nonnull components.
\end{enumerate}

The above two settings intend to investigate the impact of sparsity. Further, we consider homoscedastic and heteroscedastic regression errors for each sparsity level, where the errors are generated from six different distributions. Simulation results are presented in Tables~\ref{table: simulation test results tau0.2} and \ref{table: simulation test results tau0.6}. The test statistic in \eqref{eq: test stat} is constructed using pairs of expectile estimators $\bm\xi = (\widetilde\beta_{\tau_1}, \widetilde\beta_{\tau_2})$, where seven pairs---$\tau_2 = 0.8$, $\tau_1 = 0.1, 0.2, 0.6, $ and $\tau_2= 0.9$, $\tau_1 = 0.1, 0.2, 0.6, 0.8$---are tested. Since the test results for four pairs are similar, we only report two pairs---$(\tau_1, \tau_2) = (0.2, 0.8)$ and $(0.6, 0.8)$. Table~\ref{table: simulation test results tau0.2} uses expectile levels at $(\tau_1, \tau_2)= (0.2, 0.8)$; Table~\ref{table: simulation test results tau0.6} uses expectile levels at $(\tau_1, \tau_2) = (0.6, 0.8)$. 

    In addition, under homoscedasticity of the regression error, the test statistic in \eqref{eq: test stat} can be seen as samples from a standard normal $N(0,1)$ distribution. To confirm the theory, we also inspected the normality of the test statistic constructed using different pairs of expectile estimators. \tcr{Example QQ-plots are included in Figure~1; the examples are randomly chosen among simulation replications to avoid ``cherry-picking".}

\begin{table}[htb!]
	\centering
	\bgroup
	\def\arraystretch{1.2}
	\begin{adjustbox}{max width=\textwidth}
	\begin{tabular}{c c c c | c c c c c c} \Xhline{.8pt}
 \rowcolor[gray]{.9}
 \multicolumn{10}{c}{$\alpha = 0.05, \tau_1 = 0.2, \tau_2 = 0.8$}\\ \hline
      & \multicolumn{3}{c|}{homoscedastic} & \multicolumn{6}{c}{heteroscedastic} \\ \hline
  $\varepsilon$ & $N(0,1)$ & $t_3$ & $\mbox{Lognormal}(0,1)$ &  $N(0,1)$ & $t_3$ & $\mbox{Lognormal}(0,1)$ & $N(0,1)$ & $t_3$ & $\mbox{Lognormal}(0,1)$ \\ \hline
  & \multicolumn{3}{c|}{FP(size)} & \multicolumn{3}{c}{FP(size)} & \multicolumn{3}{c}{TP(power)}\\ \hline
  High-sparsity & 0.05 & 0.05 & 0.05 & 0.05 & 0.05 & 0.05 & 0.50 & 0.32 & 0.34\\
  Medium-sparsity & 0.06 & 0.06 & 0.05 & 0.05 & 0.06 & 0.06 &  0.09 & 0.09 & 0.09\\ \hline 
    & \multicolumn{3}{c|}{homoscedastic} & \multicolumn{6}{c}{heteroscedastic} \\ \hline
  $\varepsilon$ & mixNormal 1 & mixNormal 2 & mixNormal 3 & mixNormal 1 & mixNormal 2 & mixNormal 3 & mixNormal 1 & mixNormal 2 & mixNormal 3 \\ \hline
  & \multicolumn{3}{c|}{FP(size)} & \multicolumn{3}{c}{FP(size)} & \multicolumn{3}{c}{TP(power)}\\ \hline
  High-sparsity & 0.06 & 0.05 & 0.05 & 0.06 & 0.05 & 0.05  & 0.60 & 0.32 & 0.34\\
  Medium-sparsity & 0.06 & 0.05 & 0.06  & 0.06 & 0.05 & 0.07  & 0.19 & 0.08 & 0.10\\ \hline 
     \end{tabular}
     \end{adjustbox}
    \caption{Test ($\alpha = 0.05$) results for different distributions of $\varepsilon$ under homoscedastic and heteroscedastic variance. The pair of expectile levels is $(\tau_1, \tau_2)= (0.2,0.8)$. The top half presents results for $\varepsilon$ following $N(0,1)$, $t_3$, $\mbox{Lognormal}(0,1 )$; the bottom half for $0.9N(-0.2, 0.25) + 0.1N(1.8, 0.01)$ (mixNormal 1), $0.9N(0.2, 0.25) + 0.1N(-1.8, 0.01)$ (mixNormal 2), $0.95N(0, 0.25) + 0.05 N(0, 4)$ (mixNormal 3). Each simulation setting is replicated $R = 400$ times. The average FP proportions are calculated for homoscedastic variance, where the heteroscedastic parameter vector $\gamma_0 = \bm 0$. For heteroscedastic variance, where $\gamma_0$ has 5 nonnull components $(3, 1, -5, -5, -3)$, the average FP and TP proportions are calculated. 
	}
	\label{table: simulation test results tau0.2}
	\egroup
 \end{table}

 \begin{table}[htb!]
	\centering
	\bgroup
	\def\arraystretch{1.2}
	\begin{adjustbox}{max width=\textwidth}
	\begin{tabular}{c c c c | c c c c c c} \Xhline{.8pt}
 \rowcolor[gray]{.9}
 \multicolumn{10}{c}{$\alpha = 0.05, \tau_1 = 0.6, \tau_2 = 0.8$}\\ \hline
      & \multicolumn{3}{c|}{homoscedastic} & \multicolumn{6}{c}{heteroscedastic} \\ \hline
 $\varepsilon$ & $N(0,1)$ & $t_3$ & $\mbox{Lognormal}(0,1)$ &  $N(0,1)$ & $t_3$ & $\mbox{Lognormal}(0,1)$ & $N(0,1)$ & $t_3$ & $\mbox{Lognormal}(0,1)$  \\ \hline
  & \multicolumn{3}{c|}{FP(size)} & \multicolumn{3}{c}{FP(size)} & \multicolumn{3}{c}{TP(power)}\\ \hline
  High-sparsity & 0.05 & 0.05 & 0.05 & 0.05 & 0.05 & 0.05 & 0.24 & 0.17 & 0.17\\
  Medium-sparsity & 0.06 & 0.06 & 0.05 & 0.06 & 0.06 & 0.06 & 0.05 & 0.06 & 0.06\\ \hline 
     & \multicolumn{3}{c|}{homoscedastic} & \multicolumn{6}{c}{heteroscedastic} \\ \hline
  $\varepsilon$  & mixNormal 1 & mixNormal 2 & mixNormal 3  & mixNormal 1 & mixNormal 2 & mixnormal 3 & mixNormal 1 & mixNormal 2 & mixnormal 3\\ \hline
  & \multicolumn{3}{c|}{FP(size)} & \multicolumn{3}{c}{FP(size)} & \multicolumn{3}{c}{TP(power)}\\ \hline
  High-sparsity & 0.05 & 0.04 & 0.05 & 0.05 & 0.04 & 0.04  & 0.28 & 0.13 & 0.16\\
  Medium-sparsity & 0.06 & 0.04 & 0.04 & 0.06 & 0.04 & 0.06  & 0.10 & 0.05 & 0.09\\ \hline 
     \end{tabular}
     \end{adjustbox}
    \caption{Test ($\alpha = 0.05$) results for different distributions of $\varepsilon$ under homoscedastic and heteroscedastic variance. The pair of expectile levels is $(\tau_1, \tau_2)= (0.6,0.8)$. The top half presents results for $\varepsilon$ following $N(0,1)$, $t_3$, $\mbox{Lognormal}(0,1 )$; the bottom half for $0.9N(-0.2, 0.25) + 0.1N(1.8, 0.01)$ (mixNormal 1), $0.9N(0.2, 0.25) + 0.1N(-1.8, 0.01)$ (mixNormal 2), $0.95N(0, 0.25) + 0.05 N(0, 4)$ (mixNormal 3). ach simulation setting is replicated $R = 400$ times. The average FP proportions are calculated for homoscedastic variance, where the heteroscedastic parameter vector $\gamma_0 = \bm 0$. For heteroscedastic variance, where $\gamma_0$ has 5 nonnull components $(3, 1, -5, -5, -3)$, the average FP and TP proportions are calculated.  
	}
	\label{table: simulation test results tau0.6}
	\egroup
 \end{table}


We make several observations from Tables~\ref{table: simulation test results tau0.2} and \ref{table: simulation test results tau0.6}. 
When comparing Table~\ref{table: simulation test results tau0.2} and Table~\ref{table: simulation test results tau0.6}, we observe that using the expectile level pair $(\tau_1, \tau_2) = (0.2, 0.8)$ leads to higher test power, which agrees with Theorem~\ref{thm:power function of test} that shows the power of the test depends on the difference $(u_{\tau_2 } - u_{\tau_1})$. Placing two expectile levels far enough would preserve the signal strength of $\gamma_0$ to a certain extent. However, we do not suggest using two very extreme expectiles, since estimating them is more challenging than estimating those closer to the center.

When comparing the results for different error distributions, we note that the average $\mbox{FP}$ proportions are close to the nominal $\alpha$, which complies with Theorem~\ref{thm:power function of test}. The average $\mbox{FP}$ proportions are slightly greater than the nominal level for the mixture normal distribution $0.9N(-0.2, 0.25) + 0.1 N(1.8, 0.01)$. At the same time, this error distribution tends to have the highest empirical power among all six error distributions. 
In medium-sparsity settings, where the number of nonnull components of $\beta_0$ is 40, estimating $\beta_{\tau}$ becomes more challenging, and it is not surprising that the power of the test in this setting dropped compared to Setting~2, where the number of nonnull components of $\beta_0$ is 5. The estimation variances diminish the signal strength when $s=40$. This setting is very challenging for all tests, not just for ours.

\subsection{Check the decorrelation proposal}\label{sec:6.2}
We use simulations to check the validity of the decorrelation step proposed in Section~\ref{sec:decorrelation}. We assume an autoregressive ($AR(1)$) covariance matrix for $\bm{\tilde X} \sim N(0, \bm{\tilde \Sigma})$, where the components $\bm{\tilde \Sigma}_{ij} = \rho^{|i - j|}$ for $i, j = 1, \ldots, p$. The error distributions considered are $N(0,1)$ and $t_3$. To assess the impact of correlation on test performance, we consider $\rho = 0.3$ and $0.7$.

We see in Tables~\ref{table: simulation decorrelation 0.2} and \ref{table: simulation decorrelation 0.6} that the test is still valid numerically using the transformed data by the ``generalized puffer transformation".  We have checked that on the transformed data, the error is independent of the covariates, which suggests that the null hypothesis remains unchanged after the transformation. Simulation results demonstrate that the size of the test is approximately equal to the nominal $\alpha$. The power of the test decreases as the ``generalized puffer transformation" reduces the signal-to-noise ratio. Furthermore, the power of the test decreases when the correlation parameter $\rho$ increases from 0.3 to 0.7.

\begin{table}[htb!]
	\centering
	\bgroup
	\begin{adjustbox}{max width=\textwidth}
	\begin{tabular}{c c c | c c c c | c c | c c c c} \Xhline{.8pt}
 \rowcolor[gray]{.8}
 \multicolumn{13}{c}{$\alpha = 0.05, \tau_1 = 0.2, \tau_2 = 0.8$}\\ \hline
\rowcolor[gray]{.9} & \multicolumn{6}{c|}{$\rho = 0.3 $} &  \multicolumn{6}{c}{$\rho = 0.5$}  \\ \hline
      & \multicolumn{2}{c|}{homoscedastic} & \multicolumn{4}{c|}{heteroscedastic} & \multicolumn{2}{c|}{homoscedastic} & \multicolumn{4}{c}{heteroscedastic}  \\ \hline
 $\varepsilon$ & $N(0,1)$ & $t_3$ & $N(0,1)$ & $t_3$ & $N(0,1)$ & $t_3$ & $N(0,1)$ & $t_3$ & $N(0,1)$ & $t_3$ & $N(0,1)$ & $t_3$   \\ \hline
  & \multicolumn{2}{c|}{FP(size)} & \multicolumn{2}{c}{FP(size)} & \multicolumn{2}{c|}{TP(power)} 
  & \multicolumn{2}{c|}{FP(size)} & \multicolumn{2}{c}{FP(size)} & \multicolumn{2}{c}{TP(power)}\\ \hline
  High-sparsity & 0.05 & 0.05 & 0.05 & 0.05 & 0.40 & 0.26  & 0.05 & 0.05 & 0.05 & 0.05 & 0.33 & 0.23\\
  Medium-sparsity &  0.06 & 0.05 & 0.06 & 0.06 & 0.13 & 0.13 & 0.07 & 0.07 & 0.06 & 0.06 & 0.12 & 0.11\\ \hline
     \end{tabular}
     \end{adjustbox}
    \caption{Test ($\alpha = 0.05$) results are presented for the decorrelation procedure under both homoscedastic and heteroscedastic variance. The pair of expectile levels is $(\tau_1, \tau_2) = (0.2, 0.8)$. The covariance matrix for $\bm{\tilde X} \sim N(0, \bm{\tilde \Sigma})$ has components $\bm{\tilde \Sigma}_{ij} = \rho^{|i - j|}$ for $i, j = 1, \ldots, p$. The error distributions considered are $N(0,1)$ and $t_3$. ach simulation setting is replicated $R = 400$ times. The average FP proportions are calculated for homoscedastic variance, where the heteroscedastic parameter vector $\gamma_0 = \bm 0$. For heteroscedastic variance, where $\gamma_0$ has 5 nonnull components $(3, 1, -5, -5, -3)$, the average FP and TP proportions are calculated. 
	}
	\label{table: simulation decorrelation 0.2}
	\egroup
 \end{table}
 \begin{table}[htb!]
	\centering
	\bgroup
	\begin{adjustbox}{max width=\textwidth}
	\begin{tabular}{c c c | c c c c | c c | c c c c} \Xhline{.8pt}
 \rowcolor[gray]{.8}
 \multicolumn{13}{c}{$\alpha = 0.05, \tau_1 = 0.6, \tau_2 = 0.8$}\\ \hline
\rowcolor[gray]{.9} & \multicolumn{6}{c|}{$\rho = 0.3 $} &  \multicolumn{6}{c}{$\rho = 0.5$}  \\ \hline
      & \multicolumn{2}{c|}{homoscedastic} & \multicolumn{4}{c|}{heteroscedastic} & \multicolumn{2}{c|}{homoscedastic} & \multicolumn{4}{c}{heteroscedastic}  \\ \hline
 $\varepsilon$ & $N(0,1)$ & $t_3$ & $N(0,1)$ & $t_3$ & $N(0,1)$ & $t_3$ & $N(0,1)$ & $t_3$ & $N(0,1)$ & $t_3$ & $N(0,1)$ & $t_3$   \\ \hline
  & \multicolumn{2}{c|}{FP(size)} & \multicolumn{2}{c}{FP(size)} & \multicolumn{2}{c|}{TP(power)} 
  & \multicolumn{2}{c|}{FP(size)} & \multicolumn{2}{c}{FP(size)} & \multicolumn{2}{c}{TP(power)}\\ \hline
  High-sparsity & 0.05 & 0.06 & 0.05 & 0.06 & 0.16 & 0.12 & 0.06 & 0.05 & 0.05 & 0.05 & 0.13 & 0.10 \\
  Medium-sparsity &  0.06 & 0.05 & 0.06 & 0.06 & 0.08 & 0.09 & 0.07 & 0.07 & 0.06 & 0.07  & 0.07 & 0.08\\ \hline
     \end{tabular}
     \end{adjustbox}
    \caption{Test ($\alpha = 0.05$) results are presented for the decorrelation procedure under both homoscedastic and heteroscedastic variance. The pair of expectile levels is $(\tau_1, \tau_2) = (0.6, 0.8)$. The covariance matrix for $\bm{\tilde X} \sim N(0, \bm{\tilde \Sigma})$ has components $\bm{\tilde \Sigma}_{ij} = \rho^{|i - j|}$ for $i, j = 1, \ldots, p$. The error distributions considered are $N(0,1)$ and $t_3$. ach simulation setting is replicated $R = 400$ times. The average FP proportions are calculated for homoscedastic variance, where the heteroscedastic parameter vector $\gamma_0 = \bm 0$. For heteroscedastic variance, where $\gamma_0$ has 5 nonnull components $(3, 1, -5, -5, -3)$, the average FP and TP proportions are calculated.}
	\label{table: simulation decorrelation 0.6}
	\egroup
 \end{table}

\subsection{Real data examples}\label{ssec: real data}
\subsubsection{International economic growth data}
We apply the proposed test on the ``international economic growth" data used in \citet{belloni2011inference}, where the authors mainly focused on variable selection and estimation. We now want to test whether heteroscedasticity is present in this data. The dataset is available on the website \url{https://stuff.mit.edu/~vchern/NBER/} and consists of $n = 90$ samples. Further, we exclude the intercept term in the unprocessed dataset. The original dataset consists of 61 covariates measuring aspects such as education, trading openness, and science policies. We further expand the covariate set by including the polynomial terms up to order 4 of all 61 covariates. The expanded dataset consists of 244 covariates with a ratio $\delta$ approximately 0.37. The primary interest is the response variable---the national growth rate of GDP per capita. 
We apply the decorrelation procedure outlined in section~\ref{sec:decorrelation} before constructing the test. Similar to Section~\ref{ssec: simulation}, the test statistic is calculated using two $\ell_1$-regularized expectile estimators. We consider expectile estimators at level $\tau_2 = 0.8$ and at $\tau_1 = 0.2$. According to \citet{belloni2011inference}, the dataset is homoscedastic since no unmeasured underlying variables exist. Our test results show that none of the covariates cause heteroscedasticity across the different pairs. To assess the stability of the method when using pairs at different expectile levels, we also test $\tau_1= 0.4, 0.5, 0.6$. The test results agree with $\tau_1 = 0.2$ for this dataset.

\subsubsection{Supermarket data}
We also apply the proposed heteroscedasticity test to the ``supermarket data" used in \citet{lan2016testing}. The authors extended the $t$-test in their study and focused on testing $\beta_{0, j} \neq 0$ in a high-dimensional setting ($p \gg n$). The dataset consists of $n = 464$ daily records, where the response variable is the number of customers, and the covariates represent the sales volume of 6398 different products. Although the dataset is potentially heteroscedastic, the test proposed in \citet{lan2016testing} is robust to heteroscedastic errors.

Similarly, we utilize the fused Kolmogorov filter \citep{mai2015fused} to reduce the number of covariates to $p = 928$, achieving $\delta = n/p=0.5$. The ratio $\delta =n/p= 0.5$. Note that this reduced dimension is well beyond the theoretical limit $O(n/\log(n))$, which ensures the sure screening property of the fused Kolmogorov filter. Thus, we have high confidence that the screening step only removes noisy features and retains all important features. After applying the decorrelation procedure, we employ a pair of $\ell_1$-regularized expectile estimators at $(\tau_1, \tau_2) = (0.2, 0.8)$ to construct the test. The test results indicate that 61 covariates potentially contribute to heteroscedasticity. We also tested $\tau_1  = 0.4, 0.5, 0.6$ and $\tau_2= 0.8$. When using the pair $(\tau_1, \tau_2) = (0.6, 0.8)$, potentially 43 covariates contribute to heteroscedasticity. In addition, there are 34 common covariates detected by both $\tau_1 = 0.2$ and $\tau_1 = 0.6$. The other two pairs using $\tau_1 = 0.4, 0.6$ have more stable performance, overlapping 43 and 47 covariates with the pair using $\tau_1 = 0.2$.

\section{Concluding Remarks}\label{sec: remarks}
In this paper, we propose a high-dimensional Newey-Powell heteroscedasticity test based on contrasting expectile regression at two different expectile levels. The testing procedure is based on a bias-corrected estimator directly obtained from the AMP algorithm in Section~\ref{sec: AMP algorithm}, which is tracked by the state evolution recursion asymptotically when $n, p\to\infty$. We clarify that we want to offer an alternative debiasing framework in addition to the existing one, which requires adding a bias-correction term to the regularized estimators. Such a well-explored debiasing framework requires estimating sparse sample precision matrices and a thorough theoretical investigation for the bias-corrected estimator, which requires a considerable amount of work. The proposed test has stable performance when $\gamma_0$ is highly sparse, meaning that only a few covariates $X_j$ contribute to the heteroscedasticity. The numerical performance suggests that the error distributions also have an impact on the size and power of the test, which is worth further investigation. Although the AMP theory is developed for multiple levels, we did not explore contrasting multiple expectiles for testing heteroscedasticity in this work. This could be interesting for future research.

Upon examining Theorem~\ref{thm:power function of test}, we observe that the strength of the signal in $\gamma_0$ and the covariance of the two expectile estimators also impact the power of the test. An open question remains regarding the optimal choice of expectile levels.

While the current literature on the Approximate Message Passing (AMP) framework assumes an independent Gaussian design matrix, it is important to acknowledge that this assumption can be relaxed. However, a significant challenge in AMP remains unresolved: the lack of known theories to relax the assumption of independent covariates, accounting for both numerical and theoretical aspects. It is crucial to emphasize that this challenge is not unique to our work but is a broader issue in the AMP literature. To address this limitation in practical applications, we propose a ``generalized puffer transformation" to enhance the applicability of the proposed test to real-world data. Through a simulation study, we demonstrate that the resulting test maintains its validity by preserving its nominal size. However, a further theoretical investigation is necessary to understand and address this challenge, and it remains an open question for future research.

\section*{Acknowledgements}
The research presented in this paper was carried out on the High Performance Computing Cluster supported by the Research and Specialist Computing Support service at the University of East Anglia. Zou is supported in part by
NIH 1R01GM163244-01.

\section*{Appendix} \label{sec:appendix}
\appendix
\section{Assumptions}\label{appendix: assumptions}
For the sake of completeness, we summarize the assumptions used in our theory.
\begin{enumerate}[label=(A{{\arabic*}})]
\item 
$\varepsilon_1, \ldots, \varepsilon_n$ are i.i.d. random variables with mean zero, a finite $(2\kappa - 2)$th moment for $\kappa \ge 2$.
%

\item 
    {The empirical measure of the components of the $p$-vector $\gamma_0$, when $p\to \infty$, converges weakly to the probability measure of a random variable $\Gamma_0$ with bounded $(2\kappa-2)$th moment for $\kappa \ge 2$. }

\item 
     {The empirical measure of the components of the $p$-vector $\beta_0$, when $p\to \infty$, converges weakly to the probability measure of a random variable $B_0$} with bounded $(2\kappa-2)$th moment for $\kappa\ge 2$. 

\item 
     A standard Gaussian design: $X_i, i = 1, \ldots, n$ are i.i.d. copies of $X \in \mathbbm R^p$ with $X_{j}\sim N(0,1/n), j=1,\ldots,p$ independent and identically distributed components. 
     
\item 
     For some $\kappa>1$,
     \begin{enumerate}
         \item $\lim_{p \to \infty} \E_{\widehat{f}_{\beta_0(p)}}(B_0^{2\kappa - 2}) = \E_{f_{B_0}}(B_0^{2\kappa - 2}) < \infty$
         \item $\lim_{p \to \infty} \E_{\widehat{f}_{\gamma_0(p)}}(\Gamma_0^{2\kappa - 2}) = \E_{f_{\Gamma_0}}(\Gamma_0^{2\kappa - 2}) < \infty$
         \item $\lim_{p \to \infty} \E_{\widehat{f}_{\bm\varepsilon (p)}}(\varepsilon^{2\kappa - 2}) = \E_{f_{\varepsilon}}(\varepsilon^{2\kappa - 2}) < \infty$
         \item $\lim_{p \to \infty} \E_{\widehat{f}_{q_{k}(p)}}(B_k^{2\kappa - 2}) < \infty$.
     \end{enumerate} 
\end{enumerate}

\section{Proofs}\label{appendix: proofs}
\subsection{Proof of \eqref{eq:proximal expectile}}\label{sssec: proximal operator calculation}
\begin{proof}
    The proximal operator is defined as the minimizer of the objective function $b\rho(x) + \frac{1}{2}(x - z)^2$, which, by elementary calculus, is the root of the subgradient of the objective function. 
    The proximal operator is then obtained by solving an equation of $x$ as follows 
    \[
    b\partial\rho_{\tau}(x;u_\tau) + x - z = 0,
    \]
    where $\partial\rho_{\tau}$ is the gradient of the differentiable expectile loss function $\rho_{\tau}$ in \eqref{eq: derivative expectile loss}. 
    Calculations follow similar steps and arguments for the two cases, i.e., $x \leq u_\tau$ and $x > u_\tau$. We present the calculation steps for $x \leq u_\tau$ in detail.
    When $x \leq u_\tau$, $\partial\rho_{\tau}(x;u_\tau) = 2(1 - \tau) (x - u_\tau)$. Then, by plugging in $\partial\rho_{\tau}$ and rearranging $2b(1-\tau) (x- u_\tau) +x - z = 0$,
    we obtain \[x = \frac{z + 2b (1 - \tau)u_\tau}{2b(1-\tau) + 1}.\] The domain of $z$ is obtained by $\frac{z + 2b (1 - \tau)u_\tau}{2b(1-\tau) + 1} \leq u_\tau$ which is $z\leq u_\tau$.    
\end{proof}

\subsection{Proof of \eqref{eq: update b}}\label{sssec: update b calculation}
\begin{proof}
    By \eqref{eq:scaled eff score}, the subgradient of the effective score function $\tilde G_{\tau}(z; b)$ in \eqref{eq:effective score expectile} w.r.t. $z$ follows
    \[
     \partial_1 G_{\tau}(z; b) = \frac{\delta}{\omega}\Big\{\frac{2b(1-\tau)}{2b(1-\tau) + 1} I\{z \leq u_\tau\} + \frac{2b\tau}{2b\tau + 1} I\{z > u_\tau\}\Big\}.
    \]
    The parameter $b$ is updated by finding solutions $b >0$, s.t.
    \[
  1 = \frac{1}{n} \sum_{i = 1}^n \partial_1 G_{\tau}(z_i; b) = 
  \frac{\delta}{\omega} \Big\{\frac{1}{n}\sum_{i = 1}^n \frac{2b(1-\tau)}{2b(1-\tau) + 1} I\{z_i \leq u_\tau\} + \frac{2b\tau}{2b\tau + 1} I\{z_i > u_\tau\}\Big\}.   
    \]
\end{proof}

\subsection{Proof of Lemma~1}
\begin{proof}
    Lemma~1 generalizes \citet[Proof of Lemma~1]{bayati2011dynamics} by showing correlation properties for two sequences $k_1, k_2$. A similar proof strategy in \citet[Proof of Lemma~1]{bayati2011dynamics} is used by induction on the iteration index $t$. For $k = k_1, k_2$, let $\mathcal{B}_{k, (t)}$ be the properties for $d_{k, (t)}, q_{k,(t)}$ in \citet[Eq.~(3.15), (3.17), (3.19), (3.21), (3.23), (3.26)]{bayati2011dynamics}, and $\mathcal{H}_{k, (t+1)}$ be the properties for $m_{k, (t)}, h_{k, (t)}$ in \citet[Eq.~(3.14), (3.16), (3.18), (3.20), (3.22), (3.24), (3.25)]{bayati2011dynamics}. We assume that $\mathcal{B}_{k, (t)}$ and $\mathcal{H}_{k, (t+1)}$ hold for the chains $k_1, k_2$; this is a direct implication of Lemma~1 in \citet{bayati2011dynamics}. The proof consists of four steps
    \begin{enumerate}
        \item Assume $\mathcal{B}_{k_1, (0)}$ and $\mathcal{B}_{k_2, (0)}$ hold, then (3.27) \tcr{holds} for $t = 0$.
        \item Assume $\mathcal{H}_{k_1, (1)}$ and $\mathcal{H}_{k_2, (1)}$ hold, then (3.26) and (3.28) hold for $t = 0$.
        \item Assume $\mathcal{B}_{k_1, (t')}, \mathcal{B}_{k_2, (t')}$ hold. If (3.27) \tcr{holds} for $t' < t$, and $\mathcal{H}_{k_1, (t'')}, \mathcal{H}_{k_2, (t'')}$, (3.26), (3.28) hold for $t'' \le t$, then (3.27) \tcr{holds} for $t$.
        \item Assume $\mathcal{B}_{k_1, (t')}, \mathcal{B}_{k_2, (t')}$ hold. If (3.27) \tcr{holds} for $t' \le t$, and $\mathcal{H}_{k_1, (t'')}, \mathcal{H}_{k_2, (t'')}$, (3.26), (3.28) hold for $t'' \le t$, then (3.26) and (3.28) hold for $t+1$.
    \end{enumerate}
    Further, for $k = k_1, k_2$, we define $\mathcal{D}_{k, (t'), (t'')}$ to be the $\sigma$-algebra generated by $d_{k, (0)}, \ldots, d_{k, (t'-1)}$, $m_{k, (0)}, \ldots, m_{k, (t'-1)}$, $h_{k, (0)}, \ldots, h_{k, (t'')}$, $q_{k, (0)}, \ldots, q_{k, (t'')}$, $\bm\varepsilon$, and $\beta_k = \beta_0 + u_k \gamma_0$. Let $M_{k, (t)} = (m_{k, (0)}, \ldots, m_{k, (t-1)})$ and $Q_{k, (t)} = (q_{k, (0)}, \ldots, q_{k, (t-1)})$. We denote $m_{k, (t)}^{\|}$ the projection of $m_{k, (t)}$ onto the column space of $M_t$, which can be expressed as 
    $m_{k, (t)}^{\|} = \sum_{r = 0}^{t-1} \delta'_{k, r} m_{k, (r)}$ and $\delta'_{k, r}$'s are some fixed coefficients. Then, we define the orthogonal projection $m_{k, (t)}^\perp = m_{k, (t)} - m_{k, (t)}^{\|}$.
    Similarly, let $q_{k, (t)}^{\|}$ and $q_{k, (t)}^\perp$ be the parallel and orthogonal projections onto $Q_{k, (t)}$. 
    \begin{enumerate}
        \item Step~1. Assume $\mathcal{B}_{k_1, (0)}$ and $\mathcal{B}_{k_2, (0)}$ hold.\\
        By \citet[Proof of Lemma~1, Step~1~(a),~(b)]{bayati2011dynamics}, for $k = k_1, k_2$, let $\mathcal{D}_{k, (0), (0)}$ be generated by $\beta_k$, $q_{k, (0)}$, $\bm \varepsilon$, then 
        \begin{equation}\label{eq: d0 equivalence in distribution}
        d_{k, (0)} \mid \mathcal{D}_{k, (0), (0)} \stackrel{d}{=} \bm X q_{k, (0)}^\perp, 
        \end{equation}
        where the components $[\bm X q_{k, (0)}]_i \stackrel{d}{=} Z_k\|q_{k, (0)}\|/\sqrt{n}$. The standard normal random 2-vector $(Z_{k_1}, Z_{k_2})$ has a non-zero correlation due to the common design matrix $\bm X$.
        We first show that the following random variable satisfies the condition of the strong law of large numbers for triangular arrays, which is also stated in \citet[Theorem~3]{bayati2011dynamics}
        \[\tilde\psi_{\rm c'}(d_{k_1, (0), i}, \varepsilon_i)
      \tilde\psi_{\rm c''}(d_{k_2, (0), i}, \varepsilon_i) -
      \E_{\bm X}[\tilde\psi_{\rm c'}(d_{k_1, (0), i}, \varepsilon)
      \tilde\psi_{\rm c''}(d_{k_2, (0), i}, \varepsilon)].
      \]
      For $\rho \in (0,1)$, and we denote $c_{k, i} = ([{\bm X}q_{k, (0)}]_i, \varepsilon_i)$ and $\tilde{c}_{k, i} = ([\tilde{\bm X}q_{k, (0)}]_i, \varepsilon_i)$, $k = k_1, k_2$,
      \begin{eqnarray*}
        \lefteqn{\E\big|\tilde\psi_{\rm c'}(d_{k_1, (0), i}, \varepsilon_i)
         \tilde\psi_{\rm c''}(d_{k_2, (0), i}, \varepsilon_i) -
          \E_{\bm X}[\tilde\psi_{\rm c'}(d_{k_1, (0), i}, \varepsilon)
         \tilde\psi_{\rm c''}(d_{k_2, (0), i}, \varepsilon)]\big|^{2+\rho} }&&\\
         &=&
         \E_{\tilde{\bm X}}\Big|\tilde\psi_{\rm c'}(\tilde{c}_{k_1, i})
         \tilde\psi_{\rm c''}(\tilde{c}_{k_2, i})  - 
         \E_{\bm X}[\tilde\psi_{\rm c'}({c}_{k_1, i})
         \tilde\psi_{\rm c''}({c}_{k_2, i})\Big|^{2+\rho} \\
         &=&
          \E_{\tilde{\bm X}}\E_{\bm X}|\tilde\psi_{\rm c'}(\tilde{c}_{k_1, i})
         \tilde\psi_{\rm c''}(\tilde{c}_{k_2, i}) - \tilde\psi_{\rm c'}({c}_{k_1, i})
         \tilde\psi_{\rm c''}({c}_{k_2, i})|^{2+\rho}\\
         &=& \E |\tilde\psi_{\rm c'}(\tilde{c}_{k_1, i})\tilde\psi_{\rm c''}(\tilde{c}_{k_2, i}) - \tilde\psi_{\rm c'}(\tilde{c}_{k_1, i})\tilde\psi_{\rm c''}({c}_{k_2, i}) \\
         &&\qquad\qquad\qquad\qquad + 
        \tilde\psi_{\rm c'}(\tilde{c}_{k_1, i})\tilde\psi_{\rm c''}({c}_{k_2, i}) -
         \tilde\psi_{\rm c'}({c}_{k_1, i})\tilde\psi_{\rm c''}({c}_{k_2, i})
         |^{2+\rho} \\
         &\leq& 2^{2 + \rho}\big\{ \E|
         \tilde\psi_{\rm c'}(\tilde{c}_{k_1, i})\tilde\psi_{\rm c''}(\tilde{c}_{k_2, i}) - \tilde\psi_{\rm c'}(\tilde{c}_{k_1, i})\tilde\psi_{\rm c''}({c}_{k_2, i})
         |^{2+\rho} \\
         &&+
         \E|
         \tilde\psi_{\rm c'}(\tilde{c}_{k_1, i})\tilde\psi_{\rm c''}({c}_{k_2, i}) - \tilde\psi_{\rm c''}({c}_{k_1, i})\tilde\psi_{\rm c''}({c}_{k_1, i})
         |^{2+\rho}
         \big\}.
      \end{eqnarray*}
        The first two equalities use \eqref{eq: d0 equivalence in distribution}; $\tilde{\bm X}$ is an independent copy of $\bm X$. The inequality holds by Lemma~5. Next, we only focus on discussing the first term since the discussion for the second term follows the same arguments. 
        \begin{eqnarray}\label{eq: show SLLN applicable}
           \lefteqn{\E|
         \tilde\psi_{\rm c'}(\tilde A_{k_1, i})\tilde\psi_{\rm c''}(\tilde A_{k_2, i}) - \tilde\psi_{\rm c'}(\tilde A_{k_1, i})\tilde\psi_{\rm c''}(A_{k_2, i})
         |^{2+\rho}}&& \nonumber\\
         &\leq& \E\Big\{|\tilde\psi_{\rm c'}(\tilde A_{k_1, i})| |\tilde\psi_{\rm c''}(\tilde A_{k_2, i}) - \tilde\psi_{\rm c''}(A_{k_2, i})|\Big\}^{2+\rho} \nonumber\\
         &\leq & \E\Big\{L_1\big(1 + \|\tilde A_{k_1, i}\|^{\kappa_c'}\big) \\
         && L_2\big(
         \|[\tilde{\bm X} q_{k_2, (0)}]_i -[{\bm X} q_{k_2, (0)}]_i \|
         ) (\|\tilde A_{k_2, i}\|^{\kappa_c''-1} + \|\tilde A_{k_1, i}\|^{\kappa_c''-1} +1\big)\Big\}^{2+\rho}. \nonumber
        \end{eqnarray}
         The expectation is w.r.t. $\tilde{\bm X}$ and $\bm X$. The expectation $\E{|Xq_{k, (0)}|_i|^k} =(\frac{\|q_{k, (0)}\|}{\sqrt{n}})^k  \E|Z|^k$ is bounded by some constant for $k \geq 1$. Then, by grouping $[\tilde{\bm X} q_{k, (0)}]_i$'s and $\varepsilon_i$'s, \eqref{eq: show SLLN applicable} can further be simplified as a sum of a constant and 
         $|\varepsilon_i|^{k_c' + k_c'' -1}$ scaled by some constant. Then,   
         \[\frac{1}{n} \sum_{i =1}^n \E|
         \tilde\psi_{\rm c'}(\tilde A_{k_1, i})\tilde\psi_{\rm c''}(\tilde A_{k_2, i}) - \tilde\psi_{\rm c'}(\tilde A_{k_1, i})\tilde\psi_{\rm c''}(A_{k_2, i})
         |^{2+\rho} \leq c' n^{\rho/2}
         \] holds if 
         $\frac{1}{n} \sum_{i =1}^n |\varepsilon_i|^{(k_c' + k_c'' -1)(2+ \rho)} = \frac{1}{n} \sum_{i =1}^n |\varepsilon_i|^{(\kappa -1)(2+ \rho)} \leq c' n^{\rho/2}$ holds, which has been shown in \citet[proof of Lemma~1(b), Step~1]{bayati2011dynamics}. Then, by applying \citet[Theorem~3]{bayati2011dynamics}, we obtain 
        \begin{eqnarray*}
        \lefteqn{\lim_{n\to \infty} \frac{1}{n}\sum_{i = 1}^n \Big\{\tilde\psi_{\rm c'}(d_{k_1, (0), i}, \varepsilon_i)
      \tilde\psi_{\rm c''}(d_{k_2, (0), i}, \varepsilon_i)}&\\
      &\qquad \qquad \qquad -
      \E_{\bm X}[\tilde\psi_{\rm c'}(d_{k_1, (0), i}, \varepsilon)
      \tilde\psi_{\rm c''}(d_{k_2, (0), i}, \varepsilon)] \Big\}
      \stackrel{a.s.}{=} 0 \tcr{.}
        \end{eqnarray*}
   By taking $\psi(v) = \E_{\hat Z_{k_1}, \hat Z_{k_2}}[\tilde\psi_{\rm c'}(\hat Z_{k_1}\|q_{k_1, (0)}\|/\sqrt{n}, v)
      \tilde\psi_{\rm c''}(\hat Z_{k_2}\|q_{k_2, (0)}\|/\sqrt{n}, v)]$, where $v = \varepsilon$ in Lemma~6, and noticing $\psi(v)$ is a pseudo-Lipschitz function with order $\kappa_{c'} + \kappa_{c''} = \kappa$ by Lemma~3, we obtain 
        \begin{eqnarray*}
        \lefteqn{\lim_{n\to \infty} \frac{1}{n}\sum_{i = 1}^n \tilde\psi_{\rm c'}(d_{k_1, (0), i}, \varepsilon_i)
      \tilde\psi_{\rm c''}(d_{k_2, (0), i}, \varepsilon_i)}&\\
      &\qquad \qquad \stackrel{a.s.}{=} 
     \E_\varepsilon \E_{\hat Z_{k_1}, \hat Z_{k_2}}[\tilde\psi_{\rm c'}(\bar\sigma_{k_1, (0)}\hat Z_{k_1}, \varepsilon)
      \tilde\psi_{\rm c''}(\bar\sigma_{k_2, (0)}\hat Z_{k_2}, \varepsilon)].
        \end{eqnarray*}
    \item Step 2. Assume $\mathcal{H}_{k_1, (1)}$ and $\mathcal{H}_{k_2, (1)}$ hold.
    
    We first show that (3.28) holds. By \citet[Eq.(3.35)]{bayati2011dynamics}, it holds that for $k =k_1, k_2$
    \[
     h_{k, (1)}\mid \mathcal{D}_{k, (1), (0)} 
     \stackrel{d}{=} \tilde{\bm X} m_{k, (0)} + o_1(1) q_{k, (0)},
    \]
    where \tcr{$\tilde{\bm X} \stackrel{d}{=} \bm X$ is independent of $\mathcal{D}_{k, (1), (0)}$ and $\E[\bm X \mid \mathcal{D}_{k, (1), (0)}]$.} Further, the error term $o_1(1) q_{k, (0)}$ can be dropped in calculations, see for example \citet[proof of Lemma~1(b), (c)]{bayati2011dynamics}, \citet[proof of Lemma~1]{zhou2020detangling}. 
    We first show that the two limits $\lim_{p\to\infty}\frac{1}{p} \sum_{j = 1}^p h_{k_1, (1), j}  h_{k_2, (1), j}$ and $\lim_{n \to\infty} \frac{1}{n} \sum_{i = 1}^n m_{k_1, (0), i} m_{k_2, (0), i}$ are finite. Let $\theta_{m_{k_1, (0)}m_{k_2, (0)}}$ denote the angle between $m_{k_1, (0)}$ and $m_{k_2, (0)}$, then, with probability 1,
    \begin{eqnarray}\label{eq: lemma 1 limit finite}
          \lefteqn{\frac{1}{n} \sum_{i = 1}^n m_{k_1, (0), i} m_{k_2, (0), i} = 
          \frac{1}{n} \langle m_{k_1, (0)}, m_{k_2, (0)} \rangle } &&\\
          &=& \frac{1}{n} \|m_{k_1, (0)}\| \|m_{k_2, (0)}\|\cos \theta_{m_{k_1, (0)}m_{k_2, (0)}}
          \le  \frac{1}{n} \|m_{k_1, (0)}\| \|m_{k_2, (0)}\|  \nonumber\\
          &=&\frac{\|m_{k_1, (0)}\|}{\sqrt{n}} \frac{\|m_{k_2, (0)}\|}{\sqrt{n}} \to \bar\zeta_{k_1, (0)}\bar\zeta_{k_2, (0)}, n \to\infty, \nonumber
    \end{eqnarray}
    where $\bar\zeta_{k_1, (0)}$, $\bar\zeta_{k_2, (0)}$ are the square-root of the state evolution parameters in (3.19), (3.20) and are finite.

    Conditional on $\mathcal{D}_{k_1, (1), (0)} \vee \mathcal{D}_{k_2, (1), (0)}$,
    \begin{eqnarray}\label{eq: lemma 1 limit h}
         \lefteqn{\frac{1}{p}\sum_{j=1}^p h_{k_1, (1), j}  h_{k_2, (1), j}} && \\
    &\stackrel{d}{=}&\frac{1}{p} \sum_{j = 1}^p [\tilde{\bm X}^\top m_{k_1, (0)}]_j [\tilde{\bm X}^\top m_{k_2, (0)}]_j
    = \frac{1}{p} \langle \tilde{\bm X}^\top m_{k_1, (0)}, \tilde{\bm X}^\top m_{k_2, (0)} \rangle. \nonumber 
    \end{eqnarray}
    We first argue that $\lim_{p\to\infty}\frac{1}{p}\sum_{j=1}^p h_{k_1, (1), j}  h_{k_2, (1), j} < \infty$ almost surely.  
    \begin{eqnarray}\label{eq: limit h finite}
        \lefteqn{\frac{1}{p} \langle \tilde{\bm X}^\top m_{k_1, (0)}, \tilde{\bm X}^\top m_{k_2, (0)} \rangle }&&\\
        &\le& \frac{1}{p} \|\tilde{\bm X}^\top m_{k_1, (0)}\|\| \tilde{\bm X}^\top m_{k_2, (0)}\|
        =\frac{\|m_{k_1, (0)}\|}{\sqrt{n}}\frac{\|m_{k_2, (0)}\|}{\sqrt{n}}, p \to\infty. \nonumber
    \end{eqnarray}
    The last equality holds with probability 1 by Lemma~4 and by \eqref{eq: lemma 1 limit finite}, the above quantity is finite almost surely and converges to 
    $\bar\zeta_{k_1, (0)} \bar\zeta_{k_2, (0)}$. A similar argument is stated in \citet[Step~2 (c)]{bayati2011dynamics}.
    Next, we prove that the difference between $\frac{1}{p}\sum_{j=1}^p h_{k_1, (1), j}  h_{k_2, (1), j}$ and $\frac{1}{n} \sum_{i = 1}^n m_{k_1, (0), i} m_{k_2, (0), i}$ is negligible when $n, p\to\infty$. 
    By Lemma~8, with probability at least $1 - 6 \exp (-\frac{n}{2}(\frac{\epsilon_h^2}{2} - \frac{\epsilon_h^3}{3}))$, the following holds 
     \begin{eqnarray}\label{eq: bound h sequence}
      &&\frac{1}{p}(\frac{1+\epsilon_h}{1-\epsilon_h} \frac{\langle m_{k_1, (0)}, m_{k_2, (0)} \rangle}{\|m_{k_1, (0)}\| \|m_{k_2, (0)}\|} - \frac{2\epsilon_h}{1-\epsilon_h})\nonumber\\
       &&\leq 
       \frac{1}{p}\frac{\langle \tilde{\bm X}^\top m_{k_1, (0)}, \tilde{\bm X}^\top m_{k_2, (0)} \rangle}{\|\tilde{\bm X}^\top m_{k_1, (0)}\| \|\tilde{\bm X}^\top m_{k_2, (0)}\| } \\
       &&\leq
      \frac{1}{p}( 1 - \frac{\sqrt{1 - \epsilon_h^2}}{1 + \epsilon_h} + \frac{\epsilon_h}{1 + \epsilon_h} +
      \frac{1-\epsilon_h}{1+\epsilon_h}\frac{\langle m_{k_1, (0)} , m_{k_2, (0)} \rangle}{\|m_{k_1, (0)}\| \|m_{k_2, (0)}\|}),\nonumber
    \end{eqnarray}
    where $\epsilon_h >0$. We first address the lower bound of \eqref{eq: lemma 1 limit h}, which directly follows the lower bound of \eqref{eq: bound h sequence} as follows
    \begin{eqnarray}\label{eq: lower bound limit h}
      \lefteqn{\lim_{p\to\infty}\frac{1}{p}(\frac{1+\epsilon_h}{1-\epsilon_h} \frac{\langle m_{k_1, (0)}, m_{k_2, (0)} \rangle}{\|m_{k_1, (0)}\| \|m_{k_2, (0)}\|} - \frac{2\epsilon_h}{1-\epsilon_h})\|\tilde{\bm X}^\top m_{k_1, (0)}\| \|\tilde{\bm X}^\top m_{k_2, (0)}\|} && \nonumber \\
      &=& \lim_{p\to\infty}\bigg\{\frac{1+\epsilon_h}{1 - \epsilon_h}\frac{1}{n}\langle m_{k_1, (0)}, m_{k_2, (0)} \rangle \frac{\|\tilde{\bm X}^\top m_{k_1, (0)}\|/\sqrt{p}}{\|m_{k_1,(0)}\|/\sqrt{n}}\frac{\|\tilde{\bm X}^\top m_{k_2, (0)}\|/\sqrt{p}}{\|m_{k_2,(0)}\|/\sqrt{n}}  \nonumber \\
      && \qquad \qquad -\frac{2\epsilon_h}{1 - \epsilon_h}\frac{\|\tilde{\bm X}^\top m_{k_1, (0)}\|}{\sqrt{p}}\frac{\|\tilde{\bm X}^\top m_{k_2, (0)}\|}{\sqrt{p}} \bigg\}\nonumber\\
     &\stackrel{a.s.}{=}&  \frac{1+\epsilon_h}{1 - \epsilon_h}\frac{1}{n}\langle m_{k_1, (0)}, m_{k_2, (0)} \rangle 
     -\frac{2\epsilon_h}{1 - \epsilon_h}\frac{\|m_{k_1,(0)}\|}{\sqrt{n}}\frac{\|m_{k_2,(0)}\|}{\sqrt{n}}  ,
    \end{eqnarray}
    where the first almost sure convergence holds by Lemma~4. 
    The upper bound of \eqref{eq: lemma 1 limit h} follows similarly, and we obtain 
    \begin{eqnarray}\label{eq: upper bound limit h}
        \lefteqn{\lim_{p\to\infty}\frac{1}{p}\sum_{j=1}^p h_{k_1, (1), j}  h_{k_2, (1), j}} &&\\
      &\le& \frac{1-\epsilon_h}{1 + \epsilon_h}\frac{1}{n} \langle m_{k_1, (0)}, m_{k_2, (0)} \rangle + (1- \frac{\sqrt{1 - \epsilon_h^2}}{1 + \epsilon_h} + \frac{\epsilon_h}{1 + \epsilon_h})\frac{\|m_{k_1,(0)}\|}{\sqrt{n}}\frac{\|m_{k_2,(0)}\|}{\sqrt{n}}. \nonumber
    \end{eqnarray}
    Combining \eqref{eq: bound h sequence}, \eqref{eq: lower bound limit h}, \eqref{eq: upper bound limit h}, and let $\epsilon_h = O(n^{-1/2 + \delta})$ for some $\delta \in (0, 1/2)$, we obtain
    \begin{equation}\label{eq: proof step 2 covariance estimator convergence}
        \frac{1}{n} m_{k_1, (0), i} m_{k_2, (0), i} \stackrel{P}{\to}  \lim_{p\to\infty}\frac{1}{p}\sum_{j=1}^p h_{k_1, (1), j}  h_{k_2, (1), j} ,
    \end{equation}
    where $\stackrel{P}{\to}$ denotes convergence in probability.


The proof for (3.26) is similar to that for (3.27) and has been shown in \citet[Section~B.2.5]{zhou2020detangling} for homoscedastic linear models where $\beta_{k_1} = \beta_{k_2} = \beta_0$. We only show the proof for (3.26) in Step~2 to avoid being repetitive. Similar to Step~1, we could show that the condition in \citet[Theorem~3]{bayati2011dynamics} holds, resulting in the following convergence
\begin{eqnarray*} 
\lefteqn{\lim_{p\to \infty}\frac{1}{p}\sum_{j = 1}^p \Big\{\tilde\psi_{\rm c}(h_{k_1, (1),j}, \beta_{k_1, j})\tilde\psi_{\rm c}(h_{k_2, (1),j}, \beta_{k_2, j})}& \\
& \qquad\qquad\qquad -
\E_{\tilde{X}}[\tilde\psi_{\rm c}(h_{k_1,(1),j}, \beta_{k_1, j})\tilde\psi_{\rm c}(h_{k_2,(1),j}, \beta_{k_2, j})]\Big\} \stackrel{\rm a.s.}{=} 0.
\end{eqnarray*}

Then, using Lemma~6 for $v(u) = \beta_{0} + u \gamma_{0}$ and
\begin{eqnarray*}
   \psi(v) = \E_{\tilde{X}} \tilde\psi_{\rm c}(h_{k_1, (1), j}, v(u_{k_1}))\tilde\psi_{\rm c}(h_{k_2, (1), j}, v(u_{k_2})), 
\end{eqnarray*}
the following convergence holds
\begin{eqnarray*}
\lefteqn{ \lim_{p\to \infty}\frac{1}{p}\sum_{j = 1}^p
\E_{\tilde{X}} \Big[ \tilde\psi_{\rm c}(h_{k_1, (1), j}, \beta_{0,j}+ u_{k_1} \gamma_{0, j})\tilde\psi_{\rm c}(h_{k_2, (1), j}, \beta_{0,j}+ u_{k_2}\gamma_{0, j}) \Big] }\\
&\stackrel{\rm a.s.}{=}& \E_{B_0}\bigg[\E_{
(Z_{k_1, (0)}, Z_{k_2, (0)})}
\Big[\tilde\psi_{\rm c}(\|\frac{m_{k_r, (0)}}{\sqrt{n}} \| Z_{k_1, (0)}, B_0+u_{k_1}\Gamma_{0})\\
&&\qquad\qquad\qquad\tilde\psi_{\rm c}(\|\frac{m_{k_r, (0)}}{\sqrt{n}} \| Z_{k_2, (0)}, B_0+ u_{k_2}\Gamma_{0})
 \Big] \bigg] \tcr{.}\\
\end{eqnarray*}
 \item Step~3. Assume $\mathcal{B}_{k_1, (t')}, \mathcal{B}_{k_2, (t')}$ hold. 
 If (3.27) \tcr{holds} for $t' < t$, and $\mathcal{H}_{k_1, (t'')}, \mathcal{H}_{k_2, (t'')}$, (3.26), (3.28) hold for $t'' \le t$.\\
By \citet[Lemma~1(b) in Step~4]{bayati2011dynamics} with dropped error term, \tcr{given 
$\mathcal{D}_{k_1, (t), (t)} \vee \mathcal{D}_{k_2, (t), (t)} = \sigma(\mathcal{D}_{k_1, (t), (t)} \bigcup \mathcal{D}_{k_2, (t), (t)})$, }we have 
\begin{eqnarray*}
    \lefteqn{\tilde\psi_{c'}(d_{k_1, (0), i},\ldots, d_{k_1, (t), i}, \varepsilon_i) \tilde\psi_{c''}(d_{k_2, (0), i}, \ldots, d_{k_2, (t), i}, \varepsilon_i) 
    \stackrel{d}{=}}&\\
   &\qquad \tilde\psi_{c'}(d_{k_1, (0), i},\ldots, d_{k_1, (t-1), i}, 
    \big[\sum_{r= 0}^{t-1}\delta_{k_1, r} d_{k_1, (r)} +  \tilde{\bm X} q_{k_1, (t)}^\perp\big]_i, \varepsilon_i)\\
    &\tilde\psi_{c''}(d_{k_2, (0), i}, \ldots, d_{k_2, (t-1), i}, 
    \big[\sum_{r= 0}^{t-1}\delta_{k_2, r} d_{k_2, (r), i}+ \tilde{\bm X} q_{k_2, (t)}^\perp\big]_i,
    \varepsilon_i),
\end{eqnarray*}
where $\delta_{k, r}$'s are coefficients for $d_{k, (r)}, r = 1, \ldots, t-1$ for $k = k_1, k_2$. We define 
\[
c_{k, i}= (d_{k, (0), i},\ldots, d_{k, (t-1), i}, 
    \big[\sum_{r= 0}^{t-1}\delta_{k, r} d_{k, (r)} +  \tilde{\bm X} q_{k, (t)}^\perp\big]_i, \varepsilon_i).
\]

By \citet[Theorem~3]{bayati2011dynamics}, we obtain 
\begin{eqnarray}\label{eq: proof of main lemma SLLN}
  \lim_{n\to\infty} \frac{1}{n} \sum_{i=1}^n \Big[\tilde\psi_{c'}(c_{k_1, i})\tilde\psi_{c''}(c_{k_2, i}) - \E_{\tilde{\bm X}}\big[\tilde\psi_{c'}(c_{k_1, i})\tilde\psi_{c''}(c_{k_2, i})\big ]\Big] \stackrel{a.s.}{=} 0,
\end{eqnarray}
where $[\tilde{\bm X} q_{k, (t)}^\perp]_i \stackrel{d}{=} \hat{Z}_k \|q_{k, (t)}^\perp\|/\sqrt{n}$. Thus, \eqref{eq: proof of main lemma SLLN} can further be expressed as 
\begin{eqnarray*}
  \lim_{n\to\infty} \frac{1}{n} \sum_{i=1}^n \Big[\tilde\psi_{c'}(c_{k_1, i})\tilde\psi_{c''}(c_{k_2, i}) - \E_{(\hat Z_{k_1}, \hat Z_{k_2})}\big[\tilde\psi_{c'}(\tilde c_{k_1, i})\tilde\psi_{c''}(\tilde c_{k_2, i})\big ]\Big] \stackrel{a.s.}{=} 0,
\end{eqnarray*}
where $\tilde{c}_{k, i} =(d_{k, (0), i}, \ldots, d_{k, (t-1), i}, [\sum_{r = 0}^{t-1}\delta_{k, r} d_{k, (r)} + \hat Z_{k}\|q^\perp_{k, (t)}\|/\sqrt{n}]_i), k = k_1, k_2$.
And by \citet[proof of Lemma~1(b) for Step 3]{bayati2011dynamics}, we know that $d_{k, (t), i}$'s are Gaussian distributed random variables with variance $\bar\sigma_{k, (t)}$. Then, the following convergence holds
\begin{eqnarray*}
   \lefteqn{\lim_{n\to\infty} \frac{1}{n} \sum_{i=1}^n \Big[\tilde\psi_{c'}(c_{k_1, i})\tilde\psi_{c''}(c_{k_2, i})}& \\
   &- \E_{\hat Z}\big[\tilde\psi_{c'}(\bar\sigma_{k_1, (0)}\hat Z_{k_1, (0)}, \bar\sigma_{k_1, (t)} \hat Z_{k_1, (t)}, \varepsilon)\tilde\psi_{c''}(\bar\sigma_{k_2, (0)}\hat Z_{k_2, (0)}, \bar\sigma_{k_2, (t)} \hat Z_{k_2, (t)}, \varepsilon)\big ]\Big] \stackrel{a.s.}{=} 0.
\end{eqnarray*}
Lastly, by applying Lemma~6, we complete the proof showing (3.27) holds for iteration $t$.
  \item Step~4. Assume $\mathcal{B}_{k_1, (t')}, \mathcal{B}_{k_2, (t')}$ hold. If (3.27) \tcr{holds} for $t' \le t$, and $\mathcal{H}_{k_1, (t'')}, \mathcal{H}_{k_2, (t'')}$, (3.26), (3.28) hold for $t'' \le t$. \\
  Since the proof for (3.26) in Step~4 is similar to that for (3.27), and the complete proof for the homoscedastic regression models has been presented in \citet[Proof of Lemma~1]{zhou2020detangling}, the readers are suggested to refer to Step~2 for the main compositions of the proof. Next, we focus on presenting the detailed proof for (3.28). 
  Given $\mathcal{D}_{k_1, (t+1), (t)} \vee \mathcal{D}_{k_2, (t+1), (t)}$, and by \citet[Lemma~1(a), (b) in Step~4]{bayati2011dynamics}, we have
  \begin{equation}\label{eq: equal h in distribution}
  h_{k, (t+1)} \stackrel{d}{=} \sum_{r=0}^t\delta'_{k, (r)} h_{k, (r)} + \tilde{\bm X}^\top m_{k, (t)}^\perp, k = k_1, k_2.
  \end{equation}
By induction, for $t_1, t_2 <t$, the following holds 
\[
    \frac{1}{n} \sum_{i = 1}^n m_{k_1, (t_1), i} m_{k_2, (t_2), i} \stackrel{P}{\to}\lim_{p\to\infty} \frac{1}{p} \sum_{j=1}^p h_{k_1, (t_1+1), j} h_{k_2, (t_2+1), j}.
\]
Further, the correlation between $m_{k_1, (t)}$ and $m_{k_2, (t)}$ is mainly caused by the common components $\beta_0$, $\gamma_0$, $\bm\varepsilon$, and $\bm{X}$ which enters $m^{\|}_{k, (t)} = 
\sum_{r = 0}^{t - 1}\delta'_{k, r} m_{k, (r)}$ at $t= 0$. In addition, at each iteration step $t = 1, \ldots$, a common independent Gaussian random matrix $\bm{\tilde X}$ brings in correlation between $m_{k_1, (t)}^\perp = m_{k_1, (t)} - m_{k_1, (t)}^{\|}$ and $m_{k_2, (t)}^\perp = m_{k_2, (t)} - m_{k_2, (t)}^{\|}$. Thus, we argue the validity of the following independence, i.e., $m_{k_1, (t)}^\perp \perp m_{k_2, (t'')}$ and $m_{k_2, (t)}^\perp \perp m_{k_1, (t'')}$ for $t'' < t$, since $m_{k, (t'')}$ is a vector forming $m^{\|}_{k, (t)}$.

Then, we show that (3.28) holds for $t_1 = t$ and $t_2 < t$. By induction hypothesis and the

\begin{eqnarray*}
    \lefteqn{ \frac{1}{n}\sum_{i = 1}^n m_{k_1, (t), i}  m_{k_2, (t_2), i} }&&\\
    &=&\frac{1}{n}\sum_{i = 1}^n m_{k_1, (t), i}^{\|}  m_{k_2, (t_2), i} = \sum_{r = 0}^t \delta'_{k_1, (r)} 
      \frac{1}{n}\sum_{i = 1}^n m_{k_1, (r), i}  m_{k_2, (t_2), i}\\
     &\stackrel{P}{\to}& \sum_{r = 0}^t \delta'_{k_1, (r)} \lim_{p\to\infty} \frac{1}{p}\sum_{j = 1}^p h_{k_1, (r+1), j}   h_{k_2, (t_2+1), j} \\
     &= & 
    \lim_{p \to \infty} \frac{1}{p}\sum_{j=1}^p \big[\sum_{r = 0}^t \delta'_{k_1, (r)}h_{k_1, (r + 1)} \big]_j  \big[ h_{k_2, (t_2 + 1)} \big]_j \\
    && +  \lim_{p \to \infty} \frac{1}{p}\sum_{j=1}^p [\tilde{\bm X}^\top m_{k_1, (t)}^\perp \big]_j\big[ h_{k_2, (t_2 + 1)} \big]_j \\
     &\stackrel{d}{=}&
     \lim_{p \to \infty} \frac{1}{p}\sum_{j=1}^p \big[\sum_{r = 0}^t \delta'_{k_1, (r)} h_{k_1, (r + 1)} + \tilde{\bm X}^\top m_{k_1, (t)}^\perp\big]_j  \big[h_{k_2, (t_2 + 1)} \big]_j \\
    &=& \lim_{p \to \infty} \frac{1}{p}\sum_{j=1}^p h_{k_1, (t+1), j} h_{k_2, (t_2+1), j}.
\end{eqnarray*}
The first equality holds by the independence of $m_{k_1, (t)}^\perp$ and $m_{k_2, (t_2)}$ for $t_2 < t$; the second equality holds by the definition of $m_{k, (t)}^{\|}$; the convergence in probability holds by induction hypothesis; the fourth equality holds since the second term is null due to the independence of $\tilde{\bm X}^\top m_{k_1, (t)}^\perp$ and $h_{k_2, (t_2+1)}$ (see the argument below \eqref{eq: equal h in distribution} for detail); the equality in distribution is obtained by \eqref{eq: equal h in distribution}.


To complete the proof, we now show that (3.28) holds for $t_1, t_2  = t$. The proof technique combines \eqref{eq: proof step 2 covariance estimator convergence} and that for $t_1 = t, t_2 < t$. 
\begin{eqnarray}\label{eq: step 4 limit of h}
   \lefteqn{\lim_{p \to \infty} \frac{1}{p}\sum_{j=1}^p h_{k_1, (t+1), j} h_{k_2, (t+1), j}}&&\\
   &\stackrel{d}{=}&
     \lim_{p \to \infty} \frac{1}{p}\sum_{j=1}^p \big[\sum_{r = 0}^t \delta'_{k_1, (r)} h_{k_1, (r + 1)} + \tilde{\bm X}^\top m_{k_1, (t)}^\perp\big]_j  \big[\sum_{r = 0}^t \delta'_{k_2, (r)} h_{k_2, (r + 1)} + \tilde{\bm X}^\top m_{k_2, (t)}^\perp\big]_j  \nonumber\\
     &= & 
    \lim_{p \to \infty} \Big\{ \sum_{r, r' = 0}^t \delta'_{k_1, (r)}\delta'_{k_2, (r')} \frac{1}{p}\sum_{j=1}^p h_{k_1, (r + 1), j}  h_{k_2, (r'+ 1),j}  \nonumber\\
   &&\qquad\qquad\qquad\qquad+  \frac{1}{p}\sum_{j=1}^p [\tilde{\bm X}^\top m^\perp_{k_1, (t)}]_j [\tilde{\bm X}^\top m^\perp_{k_2, (t)}]_j 
    \Big\}. \nonumber
\end{eqnarray}
The first equivalence in distribution is by \eqref{eq: equal h in distribution}. The second and third terms on the right-hand side of the second equality vanished in the third equality due to the independence of $\tilde{\bm X}^\top m_{k, (t)}^\perp$ and $h_{k', (r+1)}$, $k, k' = k_1, k_2$, $r=0, \ldots, t$ (see the argument below \eqref{eq: equal h in distribution} for detail).
We now address the sequence
\begin{eqnarray}\label{eq: step 4 sequence m}
    \lefteqn{ \frac{1}{n}\sum_{i = 1}^n m_{k_1, (t), i} m_{k_2, (t), i}}&&\\
   &=&  \frac{1}{n}\sum_{i = 1}^n m_{k_1, (t), i}^{\|} m_{k_2, (t), i}^{\|} + \frac{1}{n} \sum_{i = 1}^n m_{k_1, (t), i}^\perp m_{k_2, (t), i}^\perp  \nonumber\\
   &\stackrel{a.s.}{=}&\sum_{r, r' = 0}^t \delta'_{k_1, (r)}\delta'_{k_2, (r')} \frac{1}{n}\sum_{i=1}^n m_{k_1, (r), i}  m_{k_2, (r'),i} \nonumber\\
   && \qquad \qquad \qquad + \lim_{p\to\infty} \frac{1}{p}\sum_{j=1}^p [\tilde{\bm X}^\top m^\perp_{k_1, (t)}]_j [\tilde{\bm X}^\top m^\perp_{k_2, (t)}]_j. \nonumber
\end{eqnarray}
The first equality holds by the definition of $m_{k, (t)}$. The first term $\frac{1}{n}\sum_{i = 1}^n m_{k_1, (t), i}^{\|} m_{k_2, (t), i}^{\|}$ in the right-hand-side of the first equality, by the definition of $m_{k, (t)}^{\|}$, can be written as \[\sum_{r, r' = 0}^t \delta'_{k_1, (r)}\delta'_{k_2, (r')} \frac{1}{n}\sum_{i=1}^n m_{k_1, (r), i}  m_{k_2, (r'),i}.\] 
The almost sure convergence, by Lemma~4, holds for 
\begin{eqnarray*}
   \frac{1}{p}\sum_{j=1}^p [\tilde{\bm X}^\top m^\perp_{k_1, (t)}]_j [\tilde{\bm X}^\top m^\perp_{k_2, (t)}]_j.
\end{eqnarray*}
By the induction hypothesis, the convergence in probability holds for the first term on the right-hand-side of the second equality of \eqref{eq: step 4 sequence m} and \eqref{eq: step 4 limit of h}. And we conclude that 
\[ \frac{1}{n}\sum_{i = 1}^n m_{k_1, (t), i} m_{k_2, (t), i} \stackrel{P}{\to}\lim_{p \to \infty} \frac{1}{p}\sum_{j=1}^p h_{k_1, (t+1), j} h_{k_2, (t+1), j}.\]
    \end{enumerate}
\end{proof}

\subsection{Proof of Lemma~\ref{lemma:normality}}\label{sssec: bivariate normality convergence}
\begin{proof}
    We first take $\tilde\psi_c$ in \eqref{eq:correlation of betas convergence} to be 
    $\tilde\psi_c(x_1,\ldots,x_{t+1}, y) = \tilde\psi(x_{t+1}, y)$. Then, \eqref{eq: weak convergence of betatilde} can be obtained by taking $\tilde\psi(x_{t+1}, y) = \psi(y -  x_{t+1})$ and using \eqref{eq: h_k}.
\end{proof}

\subsection{Proof of Theorem~\ref{thm: covariance construction}}\label{sssec: covariance construction proof}
We first prove \eqref{eq: cov exist} and \eqref{eq: cov converge}. By choosing $\widetilde\psi_c(x_1,\ldots, x_{t+1}, y) = 
- x$ in \eqref{eq:correlation of betas convergence}, we obtain
\begin{eqnarray*}
  \lefteqn{\lim_{p\to \infty} \frac{1}{p} (\widetilde\beta_{k_1, (t+1), j} - \beta_{k_1, j})(\widetilde\beta_{k_2, (t+1), j} - \beta_{k_2, j}) \stackrel{a.s.}{=}\E[(-\bar\zeta_{k_1, (t)} Z_{k_1, (t)})(- \bar\zeta_{k_2, (t)} Z_{k_2, (t)})]}&\\
   &\qquad \qquad \qquad\qquad= \bar\zeta_{k_1, (t)}\bar\zeta_{k_2, (t)} \E[Z_{k_1, (t)} Z_{k_2, (t)}] = \bar\zeta_{k_1, (t)}\bar\zeta_{k_2, (t)} \Cov(Z_{k_1}, Z_{k_2}) .
\end{eqnarray*}
The last equality holds since $(Z_{k_1, (t)}, Z_{k_2, (t)})$ is a standard normal random 2-vector. And we drop the iteration index $t$ for convenience.

Similarly, we choose $\tilde\psi_{c'}(x_{1}, \ldots, x_{t}, y)$ and $\tilde\psi_{c''}(x_{1}, \ldots, x_{t}, y)$ in \eqref{eq:correlations of error convergence} to be the rescaled effective score functions $G_{k_1}(y - x_{t}; \cdot)$ and $G_{k_2}(y - x_{t}; \cdot)$. By \eqref{eq: general recursion for AMP--residual}, we obtain
\begin{eqnarray*}
   \lefteqn{\lim_{n\to\infty} \frac{1}{n}\sum_{i = 1}^n 
   G_{k_1}(z_{k_1, i, (t)}; b_{k_1, (t)})  G_{k_2}(z_{k_2, i, (t)}; b_{k_2, (t)})}& \\
   &\qquad\qquad\qquad \stackrel{a.s.}{=} \E[G_{k_1}(\varepsilon - \bar\sigma_{k_1}\hat Z_{k_1}; b_{k_1, (t)}) G_{k_2}(\varepsilon - \bar\sigma_{k_2}\hat Z_{k_2}; b_{k_2, (t)})]\\
   & \qquad\qquad\qquad= \E[G_{k_1}(\varepsilon + \bar\sigma_{k_1}\hat Z_{k_1}; b_{k_1, (t)}) G_{k_2}(\varepsilon + \bar\sigma_{k_2}\hat Z_{k_2}; b_{k_2, (t)})] .
\end{eqnarray*}
The last equality holds by $Z \stackrel{d}{= } - Z$ for $Z\sim N(0,1)$.

To obtain \eqref{eq: cov est}, we first combine \eqref{eq: general recursion for AMP--residual}, \eqref{eq: h_k}, \eqref{eq: lemma 1c relevant} which leads to 
\begin{eqnarray*}
   \lefteqn{\frac{1}{n}\sum_{i = 1}^n 
   G_{k_1}(z_{k_1, i, (t)}; b_{k_1, (t)}) G_{k_2}(z_{k_2, i, (t)}; b_{k_2, (t)})}&& \\
    && \qquad\qquad \stackrel{P}{\to}\lim_{p\to \infty} \frac{1}{p} (\widetilde\beta_{k_1, (t+1), j} - \beta_{k_1, j})(\widetilde\beta_{k_2, (t+1), j} - \beta_{k_2, j}).
\end{eqnarray*}
Then, when $n, p$ are sufficiently large, we use \eqref{eq: cov est} to approximate the covariance of $\widetilde\beta_{k_1, (t+1)}$ and $\widetilde\beta_{k_2, (t+1)}$.

\subsection{Proof of Theorem~\ref{thm:power function of test}}\label{appendix: proof of convergence of power function}
\begin{proof}
\begin{eqnarray*}
   \lefteqn{\vartheta(\gamma_{0, j}) = I\{T_j \in R_j, \Gamma_0 = \gamma_{0, j}\} }&& \\
   &=& I \Big\{\frac{\Delta \bm{\tilde\xi}_j^\top }{\sqrt{\Delta \Sigma \Delta^\top}} > \Phi^{-1}(1 - \alpha/2)\Big\} + I\Big\{\frac{\Delta \bm{\tilde\xi}_j^\top }{\sqrt{\Delta \Sigma \Delta^\top}} < \Phi^{-1}(\alpha/2)\Big\}\\
   &=& I\Big\{\frac{\Delta (\bm{\tilde\xi}_j - (u_{\tau_1}, u_{\tau_2})\gamma_{0, j})^\top }{\sqrt{\Delta \Sigma \Delta^\top}} > \Phi^{-1}(1 - \alpha/2) - \gamma_{0, j}\frac{\Delta (u_{\tau_1}, u_{\tau_2})^\top}{\sqrt{\Delta \Sigma \Delta^\top}}\Big\}  \\
   &&+ I\Big\{\frac{\Delta(\bm{\tilde\xi}_j- (u_{\tau_1}, u_{\tau_2})\gamma_{0, j})^\top }{\sqrt{\Delta \Sigma \Delta^\top}} < \Phi^{-1}(\alpha/2) - \gamma_{0, j}\frac{\Delta (u_{\tau_1}, u_{\tau_2})^\top}{\sqrt{\Delta \Sigma \Delta^\top}}\Big\} \tcr{.}\\
  \end{eqnarray*}
 By Lemma~\ref{lemma:converge_cov} and define $\widetilde\Xi = (B_0 + u_{\tau_1}\Gamma_0 + \bar\zeta_{1}Z_1, B_0 + u_{\tau_2}\Gamma_0 + \bar\zeta_{2}Z_2)$, where $(\bar\zeta_1 Z_1, \bar\zeta_2 Z_2) \sim N(0, \Sigma)$, we obtain 
 \begin{eqnarray*}
   \lefteqn{\lim_{p\to\infty}\frac{1}{p} \sum_{j = 1}^p \vartheta(\gamma_{0, j}) \stackrel{a.s.}{=}}&& \\
  &&\E I\Big\{\frac{\Delta(\widetilde\Xi - (u_{\tau_1}, u_{\tau_2})\Gamma_0)^\top}{\sqrt{\Delta \Sigma \Delta^\top}} > 
  \Phi^{-1}(1 - \alpha/2) - \Gamma_0\frac{\Delta (u_{\tau_1}, u_{\tau_2})^\top}{\sqrt{\Delta \Sigma \Delta^\top}}
  \Big\}\\
  &&+ \E I\Big\{\frac{\Delta(\widetilde\Xi - (u_{\tau_1}, u_{\tau_2})\Gamma_0)^\top}{\sqrt{\Delta \Sigma \Delta^\top}} < 
  \Phi^{-1}(\alpha/2) - \Gamma_0\frac{\Delta (u_{\tau_1}, u_{\tau_2})^\top}{\sqrt{\Delta \Sigma \Delta^\top}}
  \Big\} \\
  &&= \E I \{ Z_3> 
  \Phi^{-1}(1 - \alpha/2) - \Gamma_0\frac{\Delta (u_{\tau_1}, u_{\tau_2})^\top}{\sqrt{\Delta \Sigma \Delta^\top}}
  \Big\} + \E I\{Z_3 < \Phi^{-1}(\alpha/2) - \Gamma_0\frac{\Delta (u_{\tau_1}, u_{\tau_2})^\top}{\sqrt{\Delta \Sigma \Delta^\top}}
  \Big\} \\
  &&= \E_{\Gamma_0}\Big[1 - \Phi\big(\Phi^{-1}(1 - \alpha/2) - \Gamma_0\frac{\Delta (u_{\tau_1}, u_{\tau_2})^\top}{\sqrt{\Delta \Sigma \Delta^\top}}\big) +\Phi\big(\Phi^{-1}(\alpha/2) - \Gamma_0\frac{\Delta (u_{\tau_1}, u_{\tau_2})^\top}{\sqrt{\Delta \Sigma \Delta^\top}}\big)
  \Big] \tcr{.}
    \end{eqnarray*}

\end{proof}

\section{Auxiliary definitions and lemmas}\label{appendix: auxiliary}
\begin{definition}{(Pseudo Lipschitz function)}\label{def:pseudo_lipschitz}
A function $\phi:\mathbbm R^m \to \mathbbm R$ is pseudo-Lipschitz of order $\kappa\geq 1$, if there exists a constant $L > 0$, such that for all $x, y \in \mathbbm R^m$
\begin{eqnarray*}
   |\phi(x) - \phi(y)| \leq L(1 + \| x\|^{\kappa - 1} + \|y\|^{\kappa -1})\|x - y\|. 
\end{eqnarray*}
\end{definition}

It follows that if $\phi$ is a pseudo Lipschitz function of order $\kappa$, then there exists a constant $L'$ such that for all $x\in \mathbbm R^m: |\phi(x)| \leq L' (1 + \|x\|^{\kappa})$.

\begin{lemma}\label{lemma: pseudo lipschitz}
    For any pseudo Lipschitz functions $\phi_{c'}, \phi_{c''}: \mathbbm R^m \to \mathbbm R$ of order $\kappa_{c'}$ and $\kappa_{c''}$, their product function $\phi_{c'}\phi_{c''}$ is a pseudo Lipschitz function of order $\kappa_{c'} + \kappa_{c''}$.
\end{lemma}
\subsection{Proof of Lemma~3}
\begin{proof}
For any pairs $x, y \in \mathbbm R^m$,
   \begin{eqnarray*}
  \lefteqn{|\phi_{c'}(x)\phi_{c''}(x) - \phi_{c'}(y)\phi_{c''}(y)|}&& \\
   &=& |\phi_{c'}(x)\phi_{c''}(x) - \phi_{c'}(x)\phi_{c''}(y) + \phi_{c'}(x)\phi_{c''}(y) - \phi_{c'}(y)\phi_{c''}(y)|\\
   &\leq& |\phi_{c'}(x)\phi_{c''}(x) - \phi_{c'}(x)\phi_{c''}(y)|+ 
   |\phi_{c'}(x)\phi_{c''}(y) - \phi_{c'}(y)\phi_{c''}(y)|\\
   &\leq& |\phi_{c'}(x)| |\phi_{c''}(x) - \phi_{c''}(y)| + 
   |\phi_{c''}(y)| |\phi_{c'}(x) - \phi_{c'}(y)| \\
   &\leq& 2\max\Big(L_1 \|x -y \|(1 + \|x\|^{\kappa_{c'}})(1 + \|x\|^{\kappa_{c''}-1} + \|y\|^{\kappa_{c''} - 1}), \\
   && L_2 \|x -y \|(1 + \|x\|^{\kappa_{c''}})(1 + \|x\|^{\kappa_{c'}-1} + \|y\|^{\kappa_{c'} - 1})
   \Big)\\
   &\leq & \max\Big(L'_1 \|x -y \|(1 + \|x\|^{\kappa_{c''}} + \|y\|^{\kappa_{c''}})(1 + \|x\|^{\kappa_{c'}-1} + \|y\|^{\kappa_{c'} - 1}),\\
  &&L'_2 \|x -y \|(1 + \|x\|^{\kappa_{c''}-1} + \|y\|^{\kappa_{c''}-1})(1 + \|x\|^{\kappa_{c'}} + \|y\|^{\kappa_{c'}})
  \Big)   
   \\
   &\leq& \max\Big(L'_1 \|x -y \| (1 + \|x\| + \|y\|)^{\kappa_{c''}} (1 + \|x\| + \|y\|)^{\kappa_{c'} - 1}, \\
   && L'_2 \|x -y \|(1 + \|x\| + \|y\|)^{\kappa_{c''} - 1} \Big)\\
   &\leq& L''\|x -y \| (1 + \|x\| + \|y\|)^{\kappa_{c'} + \kappa_{c''} - 1} \\
   &\leq& L'' 3^{\kappa_{c'} + \kappa_{c''} - 1} \|x -y \| (1 + \|x\|^{\kappa_{c'} + \kappa_{c''} - 1} + \|y\|^{\kappa_{c'} + \kappa_{c''} - 1} ) \tcr{.}
   \end{eqnarray*}
   The 5th and the last inequalities hold by Lemma~5.
\end{proof}

\begin{lemma}{(\citet[Lemma~2]{bayati2011dynamics})}\label{lemma: random matrix}
    For any deterministic unit vectors $u \in \mathbbm R^p$ and $v \in \mathbbm R^n$, and a random matrix $\tilde{\bm{X}}$ distributed as $\bm X \in \mathbbm{R}^{n\times p}$, the following conclusions hold
    \begin{enumerate}
        \item $v^\top \tilde{\bm{X}} u \stackrel{d}{=} Z/\sqrt{n}$, where $Z \sim N(0,1)$.
        \item $\lim_{n\to\infty} \|\tilde{\bm X} u\|^2 = 1$ almost surely.
        \item Consider a $d$-dimensional subspace $\mathbbm{S} \in \mathbbm{R}^n$ with orthonormal basis $s_1, \ldots, s_d$ and $\|s_i\|^2 = n, i = 1, \ldots, d$. Denote the matrix $\bm{D} = (s_1, \ldots, s_d)$ and the orthogonal projection onto $\mathbbm S$ as $P_{\mathbbm S}$. Then, we have $P_{\mathbbm S} \tilde{\bm{X}} u \stackrel{d}{=} \bm{D} x$ with $x \in \mathbbm R^d$ satisfying $\lim_{n\to \infty}\|x\| \stackrel{a.s.}{=}0$ and $x = o_d(1)$.
    \end{enumerate}
\end{lemma}

\begin{lemma}{(\citet[Theorem~1]{jameson2014some})}\label{lemma: polynomial_ineq}
    If $x_{i} \ge 0$ where $i = 1, \ldots, n$ and $p \ge 1$,  then
      $ \sum_{i = 1}^n x_{i}^p \leq (\sum_{i = 1}^n x_{i})^p \leq n^{p-1} \sum_{i = 1}^n x_{i}^p$.
    The reversed inequality holds for $p \in (0, 1)$
\end{lemma}

\begin{lemma}{\citet[Lemma~4]{bayati2011dynamics}}\label{lemma:lemma_1_coef_converge}
    Let $\kappa \ge 2$ and a sequence of vectors $\{v(N)\}_{N\ge 0}$ whose empirical distribution converges weakly to probability measure $f_{V}$ on $\mathbbm R$ with bounded $\kappa$th moment; additionally, assume that $\lim_{p \to \infty} \E_{\widehat f_{v}}(V^\kappa)  = \E_{f_{V}}(V^\kappa)$. Then for any pseudo Lipschitz function $\psi: \mathbbm R \to \mathbbm R$ of order $\kappa$:
    \begin{eqnarray*}
        \lim_{N\to \infty} \frac{1}{N} \sum_{j = 1}^N \psi(v_j) \stackrel{\rm a.s.}{=} \E[\psi(V)].  
    \end{eqnarray*}
\end{lemma}

\begin{lemma}\label{lemma: stability of Gaussian matrices}
   For any fixed unit vectors $u \in \mathbbm R^p$, $v \in \mathbbm R^n$ and random matrices $\tilde{\bm X}$ distributed as $\bm X \in \mathbbm R^{n\times p}$, it holds that 
 $\tilde{\bm X} u = 
       \Big(\sum_{j = 1}^p \tilde{\bm X}_{1, j} u_j, \cdots, 
      \sum_{j = 1}^p \tilde{\bm X}_{n, j} u_j \Big)^\top
  \stackrel{d}{=} \frac{1}{\sqrt{n}}( Z_{1}, \ldots, Z_{n})^\top \sim N(0, \frac{1}{n}I_n)$, 
where $Z_{i}$'s are i.i.d. samples of a standard normal $N(0,1)$ distributed random variable $Z$. 
Similarly, 
 $\tilde{\bm X}^\top v = 
       \Big(\sum_{i = 1}^n \tilde{\bm X}_{i, 1} u_i, \cdots, 
      \sum_{i = 1}^n \tilde{\bm X}_{i, p} u_i \Big)^\top
      \\
  \stackrel{d}{=} \frac{1}{\sqrt{n}}( Z'_{1}, \ldots, Z'_{p})^\top \sim N(0, \frac{1}{n}I_p)$, 
 where $Z'_{j}$'s are i.i.d. samples of a standard normal $N(0,1)$ distributed random variable $Z'$. 
\end{lemma}

\begin{lemma}{(\citet[Theorem~5]{shi2012margin})}\label{lemma: angle preserve}
Assume any $m_1, m_2 \in \mathbbm R^n$ and any random Gaussian matrix $\tilde{\bm X}$ distributed as $\tilde{\bm X}\in \mathbbm R^{n\times p}$. For any $\epsilon\in (0,1)$, if the inner product $\langle m_1, m_2 \rangle > 0$, then with probability at least 
   $1 - 6 \exp (-\frac{n}{2}(\frac{\epsilon^2}{2} - \frac{\epsilon^3}{3})) $,
the following holds
\begin{eqnarray*}
  \lefteqn{\frac{1+\epsilon}{1-\epsilon} \frac{\langle m_1 , m_2 \rangle}{\|m_1\| \|m_2\|} - \frac{2\epsilon}{1-\epsilon} \leq 
\frac{\langle \tilde{\bm X}^\top m_1, \tilde{\bm X}^\top m_2 \rangle}{\|\tilde{\bm X}^\top m_1\| \|\tilde{\bm X}^\top m_2\|}} &&\\
&&\qquad \qquad \leq
1 - \frac{\sqrt{1 - \epsilon^2}}{1 + \epsilon} + \frac{\epsilon}{1 + \epsilon} +
\frac{1-\epsilon}{1+\epsilon}\frac{\langle m_1 , m_2 \rangle}{\|m_1\| \|m_2\|} .
\end{eqnarray*}
\end{lemma}

\subsection{Numerical performance of the test using the estimator in (2.8)}\label{ssec: test using alternative expectile estimator}
Since estimation is more challenging when the error follows mixed normal distributions, we test the performance of the test with $u_\tau$ estimated by (2.8) in such settings. The estimator in (2.8) is only valid under the null. The performance of the test is reported in Table~\ref{table: simulation using alternative estimator of u_tau}.

 \begin{table}[H]
	\centering
	\bgroup
	\def\arraystretch{1.2}
	\begin{adjustbox}{max width=\textwidth}
	\begin{tabular}{c c c c | c c c } \Xhline{.8pt}
 \rowcolor[gray]{.9}
 $\alpha = 0.05$ & \multicolumn{3}{c}{$ \tau_1 = 0.2, \tau_2 = 0.8$} & \multicolumn{3}{c}{$\tau_1 = 0.6, \tau_2 = 0.8$} \\ \hline
  $\varepsilon$  & $N(0,1)$ & $t_3$ & mixNormal 1  & $N(0,1)$ & $t_3$ & mixnormal 1 \\ \hline
   High-sparsity & 0.05 & 0.05 &  0.06 & 0.05 & 0.05 & 0.05\\
  Medium-sparsity & 0.05 & 0.05 & 0.06 & 0.05 & 0.05 & 0.06 \\ \hline 
     \end{tabular}
     \end{adjustbox}
    \caption{Test results for $\varepsilon$ following $N(0,1)$, $t_3$, $0.9N(-0.2, 0.25) + 0.1N(1.8, 0.01)$ (mixNormal 1) under homoscedastic variance. The nominal significance level $\alpha = 0.05$. The pairs of expectile levels are $(\tau_1, \tau_2)= (0.2,0.8)$ on the left and $(\tau_1, \tau_2)= (0.6,0.8)$ on the right. Each simulation setting is replicated for $R = 400$ times, and the averaged FP proportions are calculated.
	}
	\label{table: simulation using alternative estimator of u_tau}
	\egroup
 \end{table}

\subsection{Figure~1}
\begin{figure}[hptb!]
    \centering
    \begin{minipage}[b]{.32\textwidth}
    \includegraphics[width = \textwidth, height = 3.8cm]{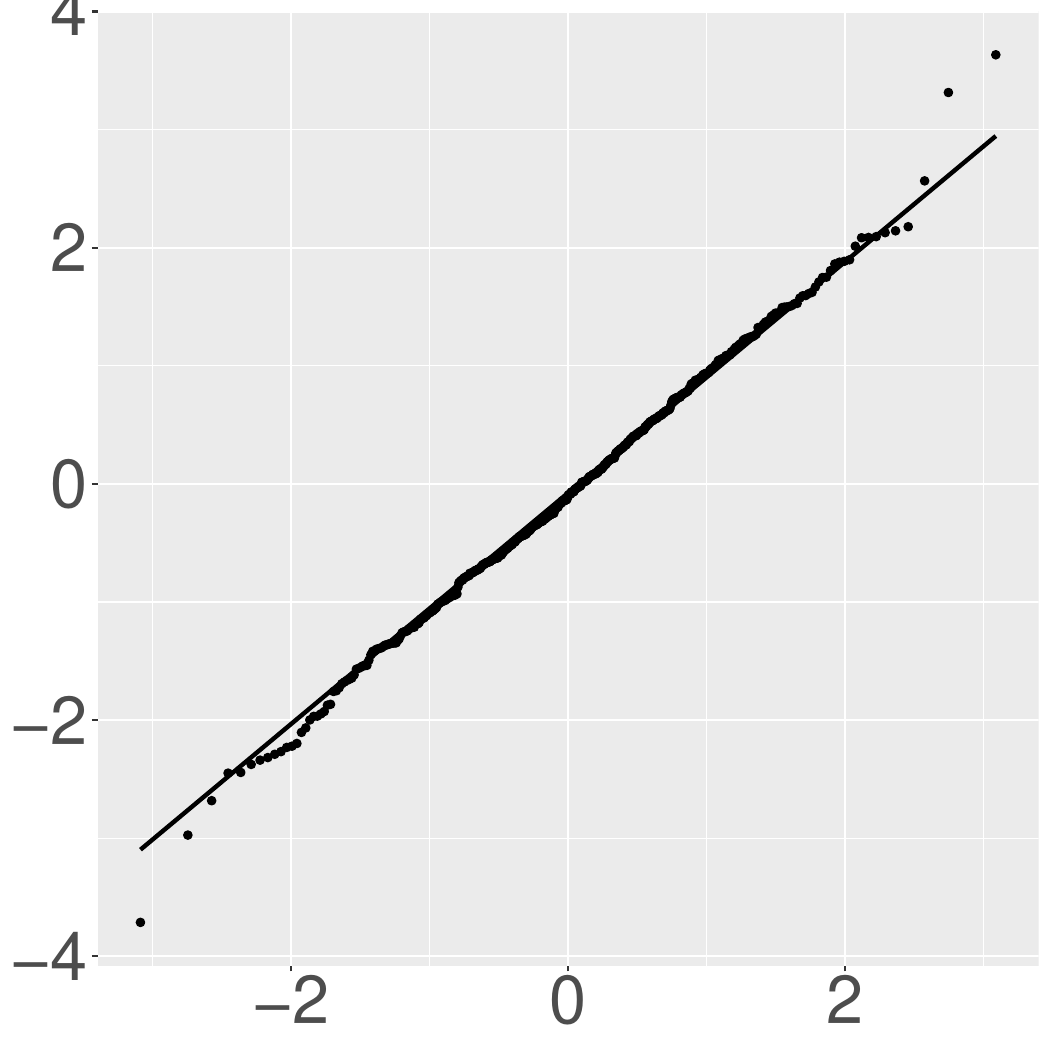}
    \subcaption{High-sparsity $s = 5$; $\tau_1 = 0.1$.}
    \end{minipage}   
    \begin{minipage}[b]{.32\textwidth}
    \includegraphics[width = \textwidth, height = 3.8cm]{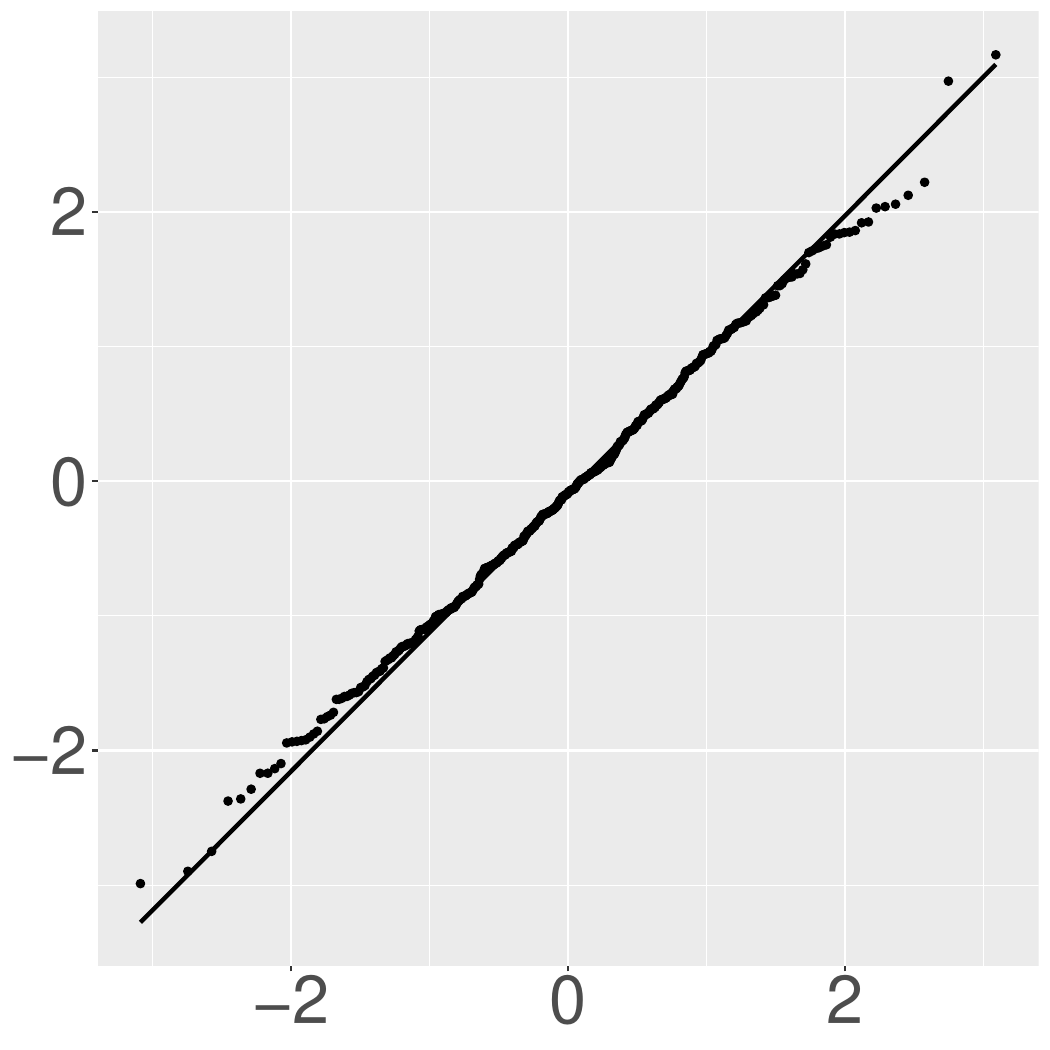}
    \subcaption{High-sparsity $s = 5$; $\tau_1 = 0.2$.}
    \end{minipage}     
    \begin{minipage}[b]{.32\textwidth}
    \includegraphics[width = \textwidth, height = 3.8cm]{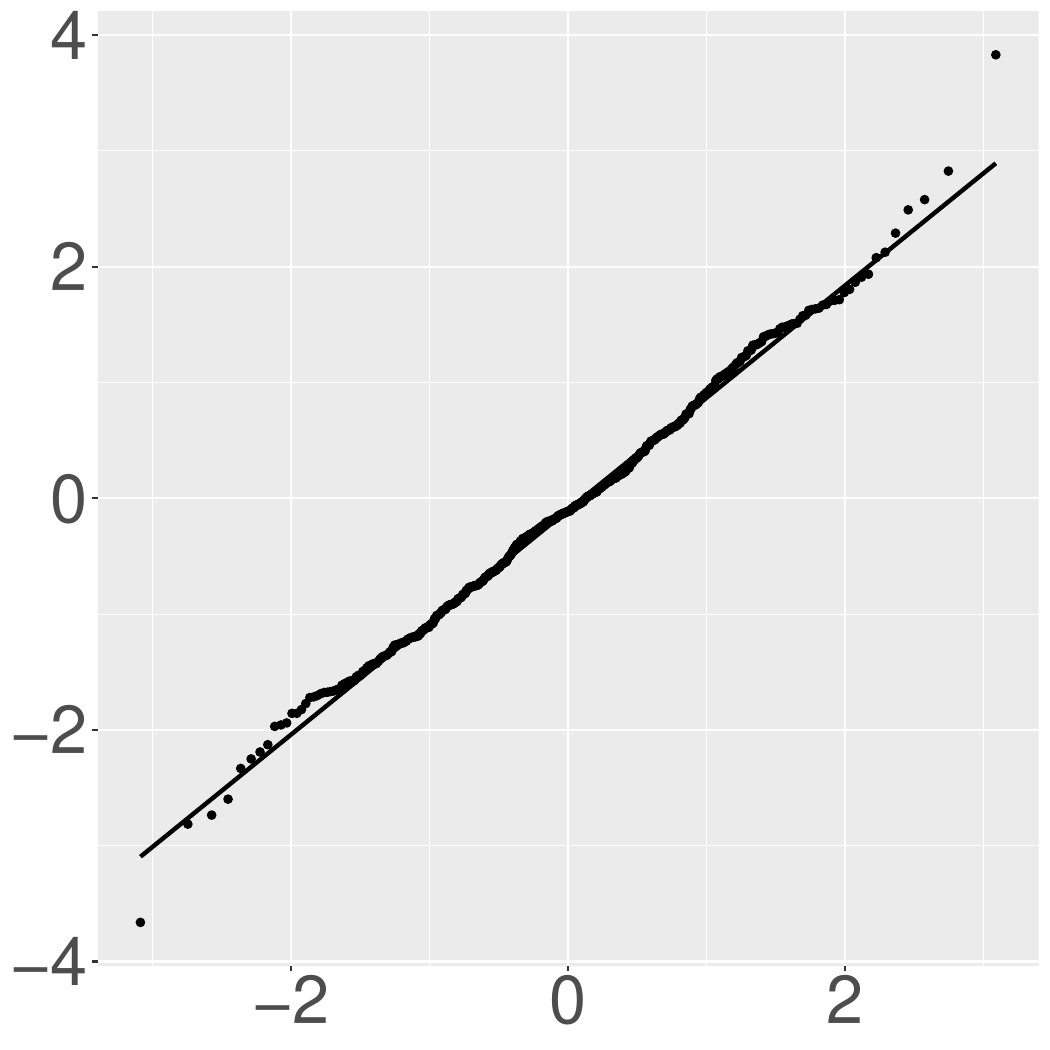}
    \subcaption{High-sparsity $s = 5$; $\tau_1 = 0.6$.}
    \end{minipage}     
    \begin{minipage}[b]{.32\textwidth}
    \includegraphics[width = \textwidth, height = 3.8cm]{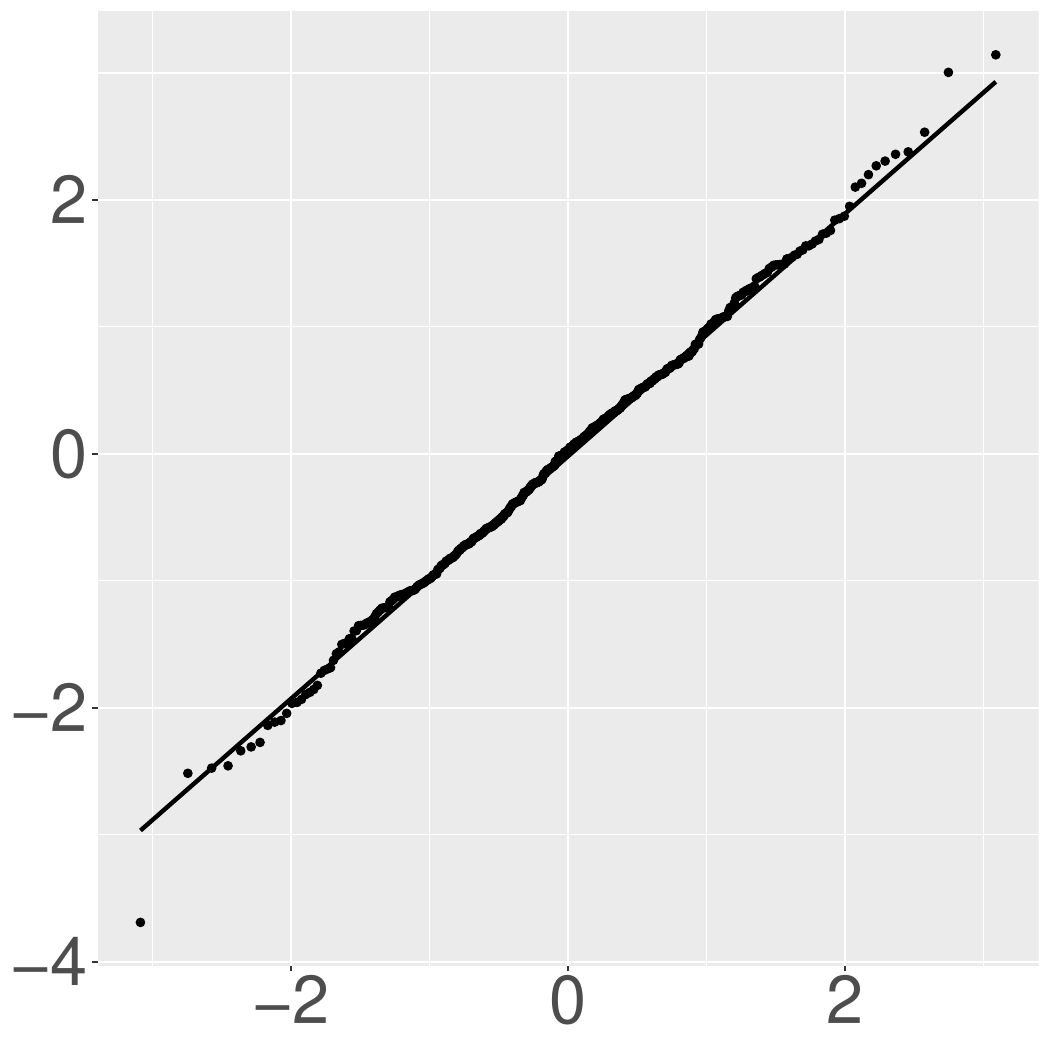}
    \subcaption{Medium-sparsity $s= 40$; $\tau_1 = 0.1$.}
    \end{minipage} 
    \begin{minipage}[b]{.32\textwidth}
    \includegraphics[width = \textwidth, height = 3.8cm]{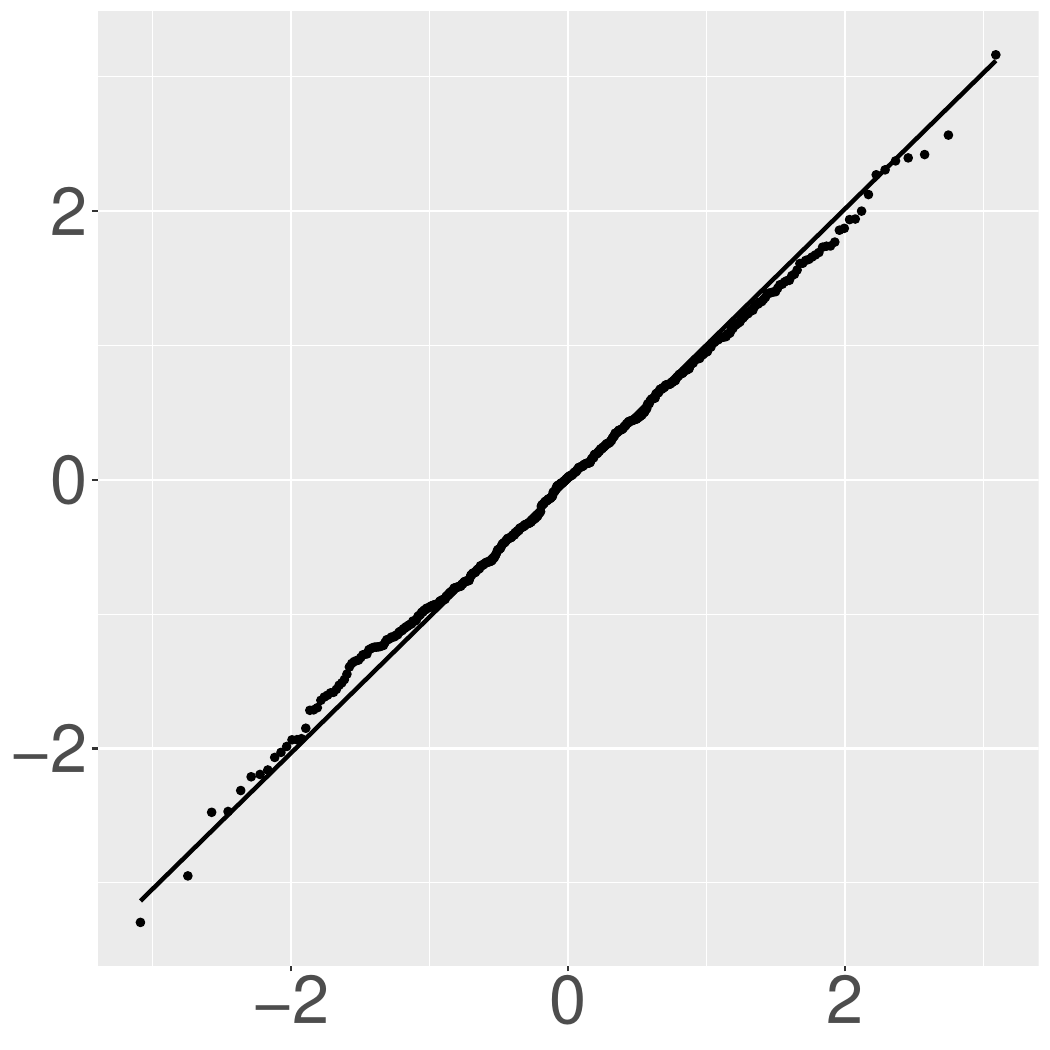}
    \subcaption{Medium-sparsity $s= 40$; $\tau_1 = 0.2$.}
    \end{minipage} 
    \begin{minipage}[b]{.32\textwidth}
    \includegraphics[width = \textwidth, height = 3.8cm]{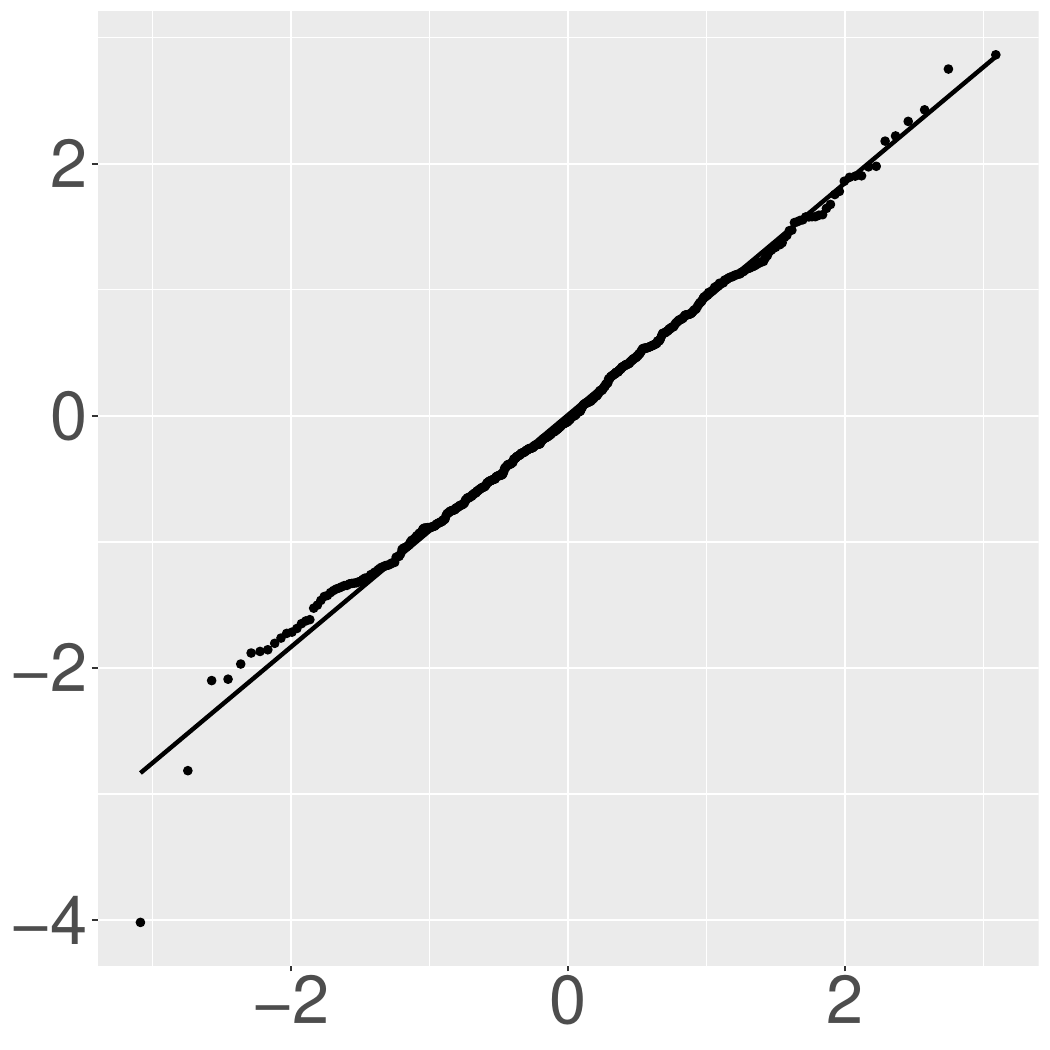}
    \subcaption{Medium-sparsity $s= 40$; $\tau_1 = 0.6$.}
    \end{minipage}
    \caption{Normal QQ-plot of test statistics $T_j$'s based on $\widetilde\beta$ at different expectile levels for three simulation settings. The test statistics $T_j$'s should approximately follow a standard normal distribution $N(0,1)$ under homoscedasticity of $\bm\varepsilon$. The expectile levels fix $\tau_2 = 0.8$ and vary $\tau_1 = 0.1, 0.2, 0.6$; The left, middle, and right columns are plots for $\tau_1= 0.1, 0.2, 0.6$, respectively. The high-sparsity and medium-sparsity plots are on the top and bottom rows, respectively.}
    \label{figure: qq plot}
\end{figure}



\bibliographystyle{apalike}

\bibliography{Expectile}   



\end{document}